%% file: draft.tex
\documentclass[12pt]{article}
\usepackage{changepage}
\usepackage{amsthm,amsmath,amssymb,amsfonts}
\usepackage{verbatim,xspace,chngcntr}

\usepackage[authoryear]{natbib}\setcitestyle{square}

\usepackage{relsize}
\usepackage{titling}
\usepackage[utf8x]{inputenc}
\usepackage{dsfont}

\usepackage[normalem]{ulem}
\usepackage{hyperref}

\hypersetup{
pdfborderstyle={},
allcolors=accentcolor,
    colorlinks=true,
    linkcolor=blue,
    citecolor=blue,
    filecolor=magenta,      
    urlcolor=cyan,
    pdftitle={I choose for you},
    pdfpagemode=FullScreen,
    }

\input{macros}

\begin{document}

\title{I Choose For You: an Experimental Study\thanks{We thank Ben Gillen for advice on econometric methodology.}}

\author{
Marina Agranov\thanks{Division of the Humanities and Social Sciences, Caltech and NBER.  Contact: \href{magranov@caltech.edu}{magranov@caltech.edu}} \hspace*{1.25em}
Federico Echenique\thanks{Department of Economics, UC Berkeley. Contact: \href{mailto:fede@econ.berkeley.edu}{fede@econ.berkeley.edu}} \hspace*{1.25em} Kota Saito\thanks{Division of the Humanities and Social Sciences, Caltech. Contact: \href{saito@caltech.edu}{saito@caltech.edu}}
}

\date{\today}

\maketitle
\begin{abstract}
    We investigate whether risk and time preferences differ when individuals make decisions for others compared to making decisions for themselves. We introduce a novel ``skin in the game'' experimental design, where choices for others incur a direct cost to the decision-maker, ensuring a genuine trade-off between self-interest and surrogate allocation. The modal outcome is that participants are more risk-averse and impatient when choosing for others than for themselves. Our methodology reveals significant heterogeneity, successfully identifying selfish types often missed by the more standard ``no skin in the game'' approaches. The message is nuanced, as even non-selfish participants behave differently when they have skin in the game. Furthermore, our framework yields more consistent behavior and superior out-of-sample predictive power.
\end{abstract}

\clearpage
\tableofcontents
\clearpage

\section{Introduction}

Economists rely heavily on understanding how individuals navigate trade-offs involving risk and time. From household savings decisions to firm investment choices, these two dimensions --- risk aversion and intertemporal substitution --- are fundamental to modeling decision-making. Typically, preferences over risk and time are elicited and analyzed under the assumption that the decision-maker is also the beneficiary of the outcome. We observe choices made by the self, for the self, and recover the underlying parameters that govern these behaviors.

A substantial portion of economic decision-making involves choices made on behalf of others. Financial advisors manage roughly \$30 trillion in U.S. household assets, making risk and timing decisions for clients; physicians select treatments for patients who defer to their expertise; and parents open investment accounts for children who cannot yet decide for themselves. In these surrogate settings, one person's preferences govern another person's consumption. Yet virtually all economic analysis of such surrogate relationships implicitly assumes that decision-makers apply the same risk and time preferences they would use for their own consumption. Whether this assumption is correct is far from obvious. This is what we do in this paper. We study whether the fundamental parameters of risk and time preference remain stable across the self–other divide, or whether the act of choosing for another induces systematically different behavior.

Specifically, we aim to recover the preferences that one subject holds over the consumption of another to determine whether there is a divergence in risk attitudes and impatience. The central empirical questions are whether individuals are more risk-averse when making choices for a third party compared to themselves and whether they exhibit different levels of impatience when scheduling consumption for others. Understanding this ``self-other wedge'' is critical for delegation and agency relationships, as it suggests that misalignments may arise not just from differing information or incentives, but also from fundamental differences in how we process risk and time for others. In other words, in analyzing choices made on behalf of others, we are focused on  misalignments that stem from differences in preferences and are more fundamental, or hard-wired, than the more common concern with conflicts of interest.

Our primary methodological contribution lies in the specific mechanism used to elicit these preferences. A standard experimental approach might attempt to study this wedge through what we call the ``no-skin-in-the-game'' design: a subject makes choices that determine the consumption of a passive recipient, but without affecting their own payoff. While this setting isolates the decision-maker’s pure preference for the other's welfare, it abstracts away from the material constraints and opportunity costs that characterize most real-world agency problems. Crucially, it is only incentivized to the extent that an experimental participant truly cares about the outcome consumed by others.

In contrast, we introduce a novel ``skin-in-the-game'' framework. In our design, the act of making choices on behalf of another is not costless; rather, it comes at the direct cost of reduced consumption for the decision-maker. This ensures that the decision-maker faces a genuine trade-off: providing a specific risk profile or time schedule for the other agent requires a sacrifice of their own resources. By linking the choices made for the beneficiary to a reduction in the decision-maker's own potential consumption, we move beyond a costless allocation. The design is explained in detail in Section~\ref{sec:expdesign}. Our approach allows us to observe how surrogate risk and time preferences manifest when the decision-maker is personally invested in the allocation process and must weigh the beneficiary's interests against their own material well-being.

Our experimental results show significant individual heterogeneity, but the most common pattern is that participants are more risk averse and more impatient with respect to others' consumption than they are with their own consumption. See Section~\ref{sec:results} for details. The skin-in-the-game methodology successfully identifies ``selfish types,'' which constitute a significant fraction of participants and allow for a more nuanced understanding of agents' preferences over others' consumption than the simple no-skin-in the-game approach. 

The comparison with the no-skin method yields some interesting conclusions. The skin in the game results are consistent across tasks and parameterizations, while the no-skin results are much less consistent. Our skin in the game estimations perform better when evaluated out of sample (one question is left out of the analysis and is used as a hold-out for predictions) than the same exercise using no-skin data. So, we conclude that the skin in the game method has some clear advantages compared to having no skin in the game. This finding is significant, as the relative attitudes towards one's own and others' consumption differ across the two methodologies, with the no-skin method underestimating how many participants are more risk averse for others than for themselves. It also underestimates how many participants have similar attitudes towards risk for themselves and others.

The difference in results between skin and no-skin cannot be explained by the mere presence of selfish agents who remain undetected in the no-skin treatment. Even non-selfish agents behave differently in the no-skin tasks compared to the skin in the game method. The incentives provided in the skin approach matter for the conclusions we draw about agents' choices on behalf of others. Given the apparent statistical benefits of the skin in the game results (stability across tasks and out-of-sample performance), we are inclined to put more weight on the conclusions from the skin, than the no-skin, treatments.

The remainder of the paper is organized as follows. We conclude this section with a review of the related literature. Section~\ref{sec:model} introduces the conceptual framework and outlines two approaches to studying how people make decisions for others: the skin-in-the-game and no-skin-in-the-game methods. Section~\ref{sec:expdesign} describes the experimental design, focusing on the identification strategy (Section~\ref{sec:Identification}) and the main features of the experimental protocol (Sections \ref{sec:Protocol} and \ref{sec:dissdesign}). Section~\ref{sec:results} presents the experimental results, beginning with non-parametric analyses (Sections~\ref{sec:SIGnonpar} and ~\ref{sec:comparenonpar}), followed by structural estimation (Section~\ref{sec:structural}) and out-of-sample predictions (Section~\ref{sec:predict}). 

\paragraph{Related literature}

The skin in the game design is inspired by the work of \cite{andreonimiller}, who developed the modified dictator game in which subjects share surplus with another participant and do so for a variety of costs of sharing and amounts available for sharing. \cite{andreonimiller} use this task to study the axioms of revealed preference in the domain of altruism. \cite{KarivEtAl} provide a graphical representation of the same task and recovers the underlying preferences for giving, taking special care to distinguish such preferences from the social preferences. \cite{KarivEtAl} find that about 25\% of people are selfish, which is in the same ballpark as our results. We extend this idea to allocating the available budget between inputs in a richer set of objects, such as the high prize in the lottery, consumed either by oneself or by others, or the delayed payments. 

Following what we have called the no-skin design,  \cite{chakravarty2011you} find that people are less risk averse when choosing on behalf of others than when they choose for themselves. The authors present two experimental designs. The first were multiple-price lists (MPLs) involving risky outcomes for the decision makers, and the second were identical MPLs for which the outcomes would be enjoyed by a randomly-drawn set of passive participants in the experiment. The subject pool consisted of management students in India. The paper documents a significantly more risk-averse behavior when choosing for self than when choosing for others. The second experimental design in the paper involves choices in first-price sealed-bid auctions and confirms qualitatively the findings from the first experiment.  To quote from the paper's conclusions ``\ldots individuals appear less risk averse when making decisions over  \other people's money, and some are actually risk loving. This general pattern is statistically significant. The difference does not appear to be driven by an attempt to pick risk attitudes that reflect the risk attitudes of others.''

\cite{reynolds2009risky} run an experiment in which each participant first chooses between a risky and a safe option for themselves; and then makes a similar choice for the benefit of a group of other people. They find that when choosing for themselves, participants are more likely to choose the risky option than when choosing on behalf of others. Thus they find, using a no-skin design, that people are more risk averse when choosing on behalf of others. 

\cite{harrison2012} consider an individual multiple-price list choice among lotteries, and a social outcome through voting. They find no significant difference in risk preferences between individual and social choices, unless the participants are provided with information about the level of risk aversion of the  other group members, in  which case they are more risk averse when choosing socially. The design is quite different from ours, though, in that the social treatment involves majority voting in groups. Interestingly, \cite{harrison2012} attribute the difference between their results and those of \cite{chakravarty2011you} to the lack of skin in the game in the latter design.

\cite{exley2016excusing} designs an experiment to understand the effect of risk on charitable giving. While the motivation and goals of her study are distinct from ours, her experimental design is probably the closest to our design. The core of Exley's experiment is a series of multiple price lists (MPL) that involve outcomes for the decision maker and for a charity. These outcomes can be certain monetary rewards or lotteries. The MPLs are arranged in four blocks. One block, which she calls ``Self Lottery vs. Charity-Certain Amount,'' is the closest to our design. The block involves two kinds of choices. First, binary choices between a lottery for the participant and a certain amount for the charity. This elicits a kind of certainty equivalent expressed in the participant's preferences for the charity's welfare. Second, a choice between a lottery for the charity and a certain amount for the charity. 

In her design, Exley varies the probability of winning in the lottery over a fine-graded scale (the set  $\{0.05, 0.10, 0.25, 0.50, 0.75, 0.90, 0.95\}$). The reason is that she is interested in the effect of uncertainty on selfishness. The main conclusion of the experiment is that risk acts as an excuse to be selfish. This is reflected in a larger aversion to donating to the charity when the lottery that the charity would obtain is subject to a larger degree of risk. While the MPL's in Exley's design are similar to ours, the use that she makes of these is orthogonal to our interest in comparing preferences over own and others' consumption.

Finally, \cite{chakraborty2025role} consider the allocation of resources among one participant who belongs to a similar social group than does the decision maker, and one who does not. The focus in their experiment is to understand the role of interpersonal uncertainty: the degree to which it is difficult for one person to assess the welfare of another. To this end, they construct alternative tasks that have no social component but that feature similar levels of uncertainty. We share with this paper the preoccupation for preferences on behalf of others, but the experimental designs, and overall goals of the two studies, are quite distinct.

\section{Model}\label{sec:model}

We first discuss our methodological proposal in an abstract setting. The methodology is broadly applicable to many different domains and questions in economics. The model primitives are $(\Theta,O,o,Q,\succeq)$ where:
\begin{itemize}
    \item $\Theta$ is a finite set of options.
    \item $O$ is a set of \df{outcomes}. 
    \item  $\succeq$ is a preference relation over $O$.
    \item $o:\Theta\times \Re_+ \to O$ is a function mapping each option $\ta$ and each \df{scale} $x\geq 0$ into an \df{outcome} $o(\ta,x)$.
    \item $Q$ is a finite set of \df{exchange rates}.
\end{itemize}

Consider two agents: \me and \other. The outcomes in $O$ will occur to (or be enjoyed by) \other. The preference $\succeq$ describes \me's choices over outcomes for \other. Suppose that \me has a strictly increasing and continuous utility $v:\Re_+\to\Re$ over \me's own consumption of money. Suppose also that $\succeq$ is represented by a utility function $w:O\to \Re$.

\begin{example}\label{ex:EU}
    Let $O$ be the set of all simple lotteries on some given set of monetary prizes: $O=\Delta_s(\Re)$ is the set of probability distributions on $\Re$ with finite support. Let $\Theta$ be a set of options, $\theta_p$ with $p=0.1,0.2,0.3$. The function $o$ assigns, to each $\theta_p$ and scalar $x\geq 0$, the lottery in $O$ that pays $x$ with probability $p$ and $0$ with probability $1-p$.

    Suppose that the utility function $w$ over lotteries in $O$ takes the expected utility form. For example, $w(\ell_p) = pu(x)+(1-p)u(0)$, when $\ell_p\in O$ is the lottery that pays $x$ with probability $p$ and $0$ with probability $1-p$. The function $u$ may, in addition, be assumed to be of the constant relative risk aversion form. So we get that $w(\ell_p) = Ap x^\gamma$, for some parameters $A>0$ and $\gamma >0$.

  \end{example}
How can we elicit, or estimate, the preferences $\succeq$ and the utility $w$? The standard procedure in economics is to ask \me to make choices in $O$ that will be consumed by \other. Essentially, preferences are choices. In our paper, we propose a different approach. It is also rooted in the notion that preferences are choices, but we offer a different channel for experimentally incentivizing those choices.

We shall impose the following assumptions on the model primitives:
\begin{enumerate}
    \item Separability: $o(\ta,x)\succeq o(\ta',x)$ if and only if $o(\ta,x')\succeq o(\ta',x')$ for any $\ta,\ta'\in\Theta$ and $x,x'>0$.
    \item $x\mapsto w(o(\ta,x))$ is a strictly monotonically increasing and continuous function.
    \item $v(1-x)\geq \max\{w(\ta,q x):\ta\in \Theta,q\in Q \}$ for all $x$ small enough, and $v(1-x)\leq \min\{w(\ta,q x):\ta\in \Theta,q\in Q \}$ for all $x$ large enough.
\end{enumerate}

Separability assumptions are extremely common in economics. We use them to define a preference on some relevant sub-domain: for example, the separation of tastes from beliefs requires separability assumptions, as does the separation of time preferences from consumption within a period. Here we impose separability in order to define a preference $\succeq^*$ over $\Theta$ by $\ta\succeq^* \ta'$ if and only if $o(\ta,x)\succeq o(\ta',x)$ for some $x>0$. We may interpret the preferences $\succeq^*$ as representing the choices that \me would make for \other to consume when the choice involves alternatives in $\Theta$.

We now discuss two approaches to recover \me's preferences over \other's consumption. The first approach, ``no skin in the game,'' is the standard procedure in the literature. The second approach, ``skin in the game,'' is our main proposal. It is important to understand the differences between the two, not only because it matters for our methodological contribution, but also for our experimental results. Our experimental design includes tasks given by both methodologies, and we are interested in the empirical results obtained under each of the two approaches.

\subsection{No Skin In The Game Method}\label{sc:modelnoSIG}

The most common experimental design asks participants to choose  between two alternatives, $o(\ta,x)$ and $o(\ta',x)$. We call this design, no-skin-in-the-game method.

In principle, such a design allows us to recover the preferences that \me has over \other's consumption. These would be the preferences $\succeq^*$ in our previous discussion. However, this method is only incentivized to the extent that \me truly cares about \other's consumption. The choices made by \me are physically paid out, and enjoyed, by \other. There is no ``physical'' connection between \other's consumption and \me's material consumption, hence the name, no-skin-in-the-game.

\subsection{Skin In The Game Method}\label{sc:modelSIG}

Our main methodological proposal is to  make \me's choices be costly for \me. Our proposal is inspired by \cite{andreonimiller}, who allow for a tradeoff between \me and \other's consumption at varying terms of trade. In our model, these terms of trade are given by the exchange rates $q\in Q$.
We ask \me to choose between \me receiving an amount of money or giving money that scales an outcome that \other gets. The exchange rate $q$ multiplies how goes into the scale ($x$) of the outcome that \other receives. 

The total amount of money is $m$, which we assume for now is $m=1$. Given $x$, then, \me chooses between getting $m-x$ or giving \other the outcome $o(\ta,q x).$ Under the assumptions we have made, the \df{cutoff} $x^*(\ta,q)$ is well defined by
\[ 
v(1-x^*(\ta,q)) = w ( o(\ta,q x^*(\ta,q)) )
\] The cutoff $x^*(\ta,q)$ is the monetary amount that equates \me's utility from \me consuming $1-x^*(\ta,q)$ with the utility that \me obtains from \other consuming the object $o(\ta,q x^*(\ta,q))$. This object could be a lottery, as in Example~\ref{ex:EU}, a dated reward, or other consumption items. 

The assumptions we have made are minimal, but they suffice for the method we have outlined to enable recovery of \me's preference for consumption by \other.

\begin{proposition}\label{prop:basic}
$\ta\succeq^* \ta'$ if and only if $x^*(\ta,q)\leq x^*(\ta',q)$.
\end{proposition}
\begin{proof}
Let $x=x^*(\ta,q)$. Then  $\ta\succeq^* \ta'$ if and only if
\[
0= v(1-x) -  w ( o(\ta,q x) )
\leq v(1-x) -  w ( o(\ta',q x) ).
\] So $x\leq x^*(\ta',q)$ as  $x\mapsto v(1-x) -  w ( o(\ta',q x) )$ is strictly decreasing.
\end{proof}

The skin-in-the-game method provides a direct connection between \me's choice of \other's consumption and \me's own consumption. Choosing an amount $x$ means two things: it means that \other consumes the object $o(\ta,qx)\in O$ and that \me consumes the monetary amount $1-x$.

Note the role played by the cutoff $x^*(\ta,q)$ in Proposition~\ref{prop:basic}. Operationally, in our experimental design, we shall use a multiple price list (MPL) to elicit a point of indifference between consumption for \me and for \other. 

\section{Experimental Design}\label{sec:expdesign}
We designed an experiment that elicits both ordinal and cardinal comparisons of participants’ preferences over outcomes consumed by themselves (\me) and by another individual (\other), allowing us to compare the SIG and no-SIG methods. As shown below, our methodology is broadly applicable. In this paper, we illustrate it in two domains that are central to the literature: risk and time preferences.

We begin by describing the types of questions our method uses and the nonparametric identification strategy for comparing preferences over outcomes for \me and  \other (Section~\ref{sec:Identification}). We then present a detailed overview of the experimental protocol (Section~\ref{sec:Protocol}) and discuss its key design features (Section~\ref{sec:dissdesign}).

\subsection{Ordinal Comparison of Preferences: Identification}\label{sec:Identification} 

Our \textbf{skin-in-the-game} method (\textbf{SIG}) centers on three core, incentive-compatible questions that pin down how people value objects consumed by either \me or \other. We recover these valuations from an indifference condition: the point at which the utility \me gets from \other consuming an object exactly matches the utility \me gets from \me consuming a complementary object. In this way, \me's utility from \other's consumption becomes directly comparable -- and commensurate -- to \me's own consumption utility.

Let $v(o)$ and $w(o)$ denote the \me's utility for outcome $o$ when consumed by \me and  \other, respectively. In the risk domain, outcomes $o \in O$ are monetary lotteries; in the time domain, they are consumption bundles specifying monetary payments delivered at different dates. For convenience, we slightly abuse notation and write $v(x,t)$ and $w(x,t)$ for the Bernoulli utility of \me for receiving a payment of \$$x$ in $t$ days by \me and  \other, respectively. For example, $v(5,0)$ denotes \me's utility of \me consuming \$5 today, while $w(5,14)$ denotes \me's utility of  \other consuming \$5 in 14 days.  

The three core questions are valuation tasks, implemented via a multiple-price list (see Section \ref{sec:Protocol}). These questions elicit three indifference points—$(x^{\ast}, y^{\ast}, z^{\ast})$—in each domain, as specified below:

\begin{itemize}
\item Risk domain
\begin{itemize}
        \item[(Q1)] $v(K-x^{\ast},0) = p \cdot v(a(x^{\ast}+c),0) + (1-p) \cdot v(0,0)$
        \item[(Q2)]$w(q(K-x^{\ast}),0) =v(y^{\ast},0)$
        \item[(Q3)] $p \cdot w(qa(x^{\ast}+c),0) + (1-p) \cdot w(0,0) = v(z^{\ast},0)$ 
\end{itemize}
\item Time domain
    \begin{itemize}
        \item[(Q1)] $v(K-x^{\ast},0) = v(a(x^{\ast}+c),t)$
        \item[(Q2)] $w(q(K-x^{\ast}),0) =v(y^{\ast},0)$
        \item[(Q3)] $ w(qa(x^{\ast}+c),t)  = v(z^{\ast},0)$ 
    \end{itemize}
\end{itemize}

The parameters $p$ and $t$ specify the main features of the studied objects: $p$ is the probability of a high prize in the risk domain, while $t$ is the delay in the time domain. The parameter $q$ scales dollar payments to \me when the same payment is received by  \other. Think of $q$ as a price. It is the exchange rate between consumption by \me and by  \other, as in \cite{andreonimiller}.\footnote{\cite{andreonimiller} vary $q$ substantially across different treatments. We do not. This reflects an important difference between their objectives and ours. The wanted to test for rationality, which requires looking at the response to price variation. We instead choose $q$ so as to make the tradeoff between \me and \other's consumption be meaningful.} Finally, the remaining parameters $(K, a, c)$ are constants calibrated to target empirically relevant regions of preferences. All parameters are fixed in advance and known to the participants. We report these parameters and discuss their role in Section \ref{sec:Protocol}.

In the risk domain, the first question, Q1, elicits the amount $x^{\ast}$ for which \me is indifferent between receiving $K-x^{\ast}$ for sure and a lottery paying $a(x^{\ast} + c)$ with probability $p$ and zero otherwise. The next two questions, Q2 and Q3, elicit $y^{\ast}$ and $z^{\ast}$, which capture \me's valuation of  \other receiving the same fixed amount $K-x^{\ast}$ and the same lottery, respectively, where other's payments are scaled by $q$.

The three values $(x^{\ast},y^{\ast},z^{\ast})$ suffice to compare \me's risk preferences for outcomes consumed by \me versus  \other. If $y^{\ast} < z^{\ast}$ then \me is more risk-averse for lotteries consumed by \me than for  \other. This follows directly from the monotonicity assumption: $y^{\ast} < z^{\ast}$ implies that $v(y^{\ast},0)<v(z^{\ast},0)$, which in turn means that \me's certainty equivalent for consuming the lottery paying $a(x^{\ast}+c)$ with probability $p$ (and zero otherwise) is smaller than \me's certainty equivalent for  \other consuming the same lottery. If, however, $y^{\ast} > z^{\ast}$ then \me is less risk-averse for lotteries consumed by \me than for  \other.

Similarly, in the time domain, Q1 elicits the amount $x^{\ast}$ for which \me is indifferent between receiving $K-x^{\ast}$ today and $a(x^{\ast} + c)$ in $t$ days. Questions Q2 and Q3 then elicit $y^{\ast}$ and $z^{\ast}$, which capture \me's valuation of  \other receiving the same amounts today and in $t$ days, respectively. As always, both payments are scaled by the exchange rate $q$. If $y^{\ast} < z^{\ast}$, then \me is more impatient for delays in own payments compared to other's payments, since $y^{\ast} < z^{\ast}$ implies that $v(y^{\ast},0)<v(z^{\ast},0)$, which in turn indicates that the certainty equivalent of \me receiving $a(x^{\ast}+c)$ in $t$ days is smaller than certainty equivalent of  \other receiving the same delayed amount scaled by $q$. Conversely, if $y^{\ast} > z^{\ast}$, \me is more patient for delays in own payments than in other's payments.

The \textbf{no-skin-in-the-game} method, the \textbf{no-SIG} hereafter, instead relies on the comparison of the values elicited using Q1 (specified above) and one additional question:
\begin{itemize}
\item Risk domain
\begin{itemize}
        \item[(Q4)] $p \cdot w(a(x^{\ast}+c),0) + (1-p) \cdot w(0,0) = w(m^{\ast},0)$ 
\end{itemize}
\item Time domain
    \begin{itemize}
        \item[(Q4)] $ w(a(x^{\ast}+c),t)  = w(m^{\ast},0)$ 
    \end{itemize}
\end{itemize}

In the risk domain, Q4 elicits the certainty equivalent of the same lottery as in Q1, but consumed by  \other and expressed in terms of \other's sure amount $m^{\ast}$. An important feature of Q4 is that the decision-maker (\me) focuses entirely on choosing outcomes for  \other, without any reference to their own outcomes. This differs from Q2 and Q3, where \me evaluates objects for  \other relative to objects consumed by \me. Using $(x^{\ast},m^{\ast})$, we can compare the ordinal risk preferences of \me for outcomes consumed by \me versus \other: if $K-x^{\ast}<m^{\ast}$, then \me is more risk-averse for lotteries consumed by \me than for those consumed by  \other, since the certainty equivalent of the same lottery in \me's sure amount is smaller than in \other's sure amount. Similarly, in the time domain, $K-x^{\ast}<m^{\ast}$ indicates that \me is more impatient for delays in their own payments than in \other's payments.

In addition to the questions Q1-Q4 we have described above, our experiments included one more question:
\begin{itemize}
\item Risk domain
\begin{itemize}
        \item[(Q5)] $v(K-s^{\ast},0) = p \cdot w(qa(s^{\ast}+c),0) + (1-p) \cdot w(0,0)$
\end{itemize}
\item Time domain
    \begin{itemize}
        \item[(Q5)] $v(K-s^{\ast},0) = w(qa(s^{\ast}+c),t)$ 
    \end{itemize}
\end{itemize}
 
In the risk domain, Q5 elicits the amount $s^{\ast}$ that makes \me indifferent between receiving $K-s^{\ast}$ for sure and \other receiving the lottery that pays 
$a(s^{\ast}+c)$ with probability $p$. In the time domain, Q5 elicits the amount $s^{\ast}$ that makes \me indifferent between receiving $K-s^{\ast}$ today and the  \other person receiving $a(s^{\ast}+c)$ in t days. Q5 was primarily used for the structural estimation of preference parameters, as discussed in Section~\ref{sec:structural}.

\subsection{Experimental Protocol}\label{sec:Protocol}

The experiment consisted of 20 rounds and 20 MPLs, one per round. Each MPL consists of a sequence of ``left'' and ``right'' option. In each row of the MPL, the participant has to choose between the left and the right options. These are structured so as to find a point of indifference between the left and right options. At the end of the experiment we included several control tasks. We summarize the main features of the experiment. The Online Appendix features the instructions that were given to our participants, as well as screenshots from the experimental interface.

\paragraph{Valuation tasks.} MPLs were used to elicit the indifference points  $(x^{\ast},y^{\ast},z^{\ast},m^{\ast},s^{\ast})$. Each round featured one MPL, which varied across two dimensions: (i) whether the left option remained fixed (a standard MPL) or varied with the right option (a modified MPL), and (ii) the identity of the recipient of the left and right outcomes (\me–\me, \other–\me, \me–\other, \other–\other).

A \textbf{standard MPL} consists of a sequence of rows, each displaying two options (left and right). The left option remains fixed across rows, while the right option changes, becoming progressively more attractive as one moves down the list. An illustration is provided in the left panel of Figure~\ref{fig:MPLs}, which shows question Q2 in the risk domain. Subjects were asked to indicate the point at which they switched from preferring one option to the other, thereby enforcing a single switching point consistent with our theoretical framework, in which individuals reveal a unique indifference point.\footnote{This is a standard experimental technique used most recently in \cite{NielsenRehbeck} and \cite{McGranaghanEtAl2024}.} To implement this, clicking the left option in a given row automatically selected the left option for that row and all rows above it, and the right option for all rows below; clicking the right option worked analogously. Participants could revise their selection at any time before submitting their final choice for each list.

A \textbf{modified MPL} also consists of a sequence of rows with two options in each; but in this case both the left and right options vary across rows: the left option becomes progressively less attractive, while the right option becomes progressively more attractive. An illustration is provided in the right panel of Figure \ref{fig:MPLs}, which shows question Q1 in the time domain. As with the standard MPL, subjects were asked to indicate the point at which they switched from preferring one option to the  \other, thereby enforcing a single switching point consistent with our theoretical framework.

The modified MPL is meant to mimic the design in \cite{andreonimiller}. In their design, as participants evaluate two choices on a budget line, they are simultaneously considering an increase in \me's consumption and a decrease in \other's. This tradeoff is obtained in our modified MPL by making an option for \me more attractive at the same time as the outcome for \other becomes worse. 

\begin{figure}
    \begin{center}
    \caption{Two types of multiple-price lists}
    \label{fig:MPLs}
    {\footnotesize standard MPL} \quad \quad \quad \quad \quad \quad \quad \quad  \quad \quad \quad {\footnotesize modified MPL} \\
    \includegraphics[scale=0.35]{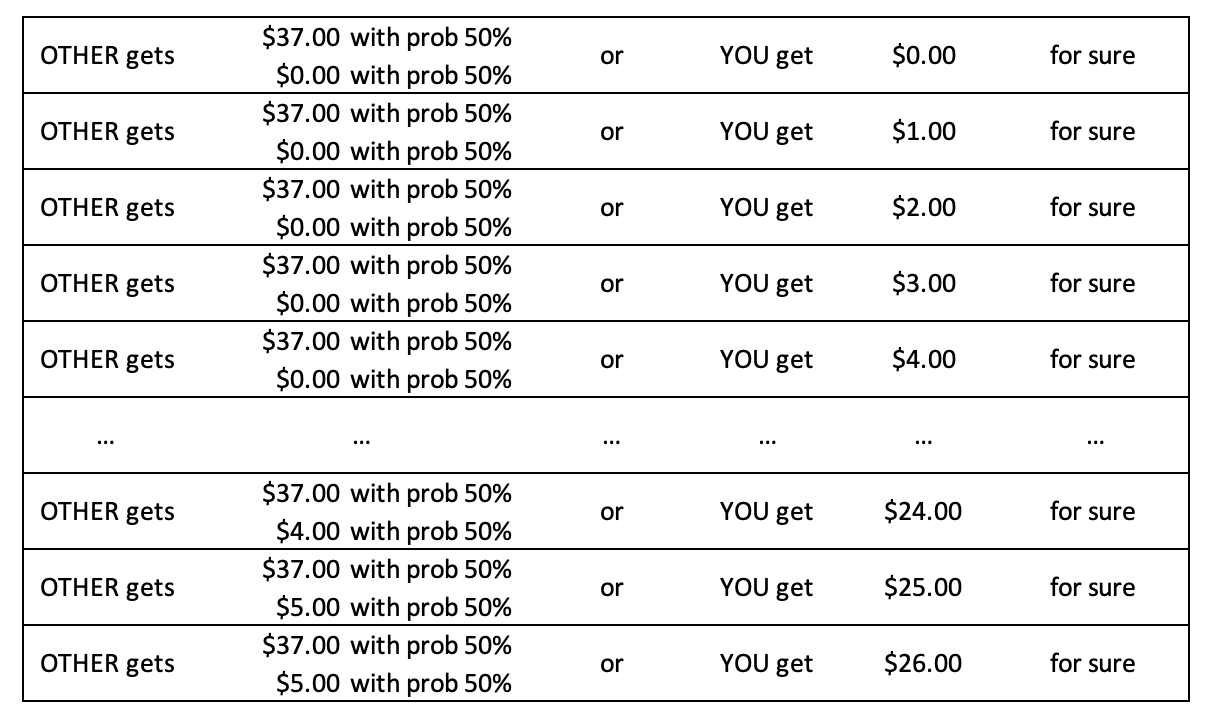} \quad
    \includegraphics[scale=0.53]{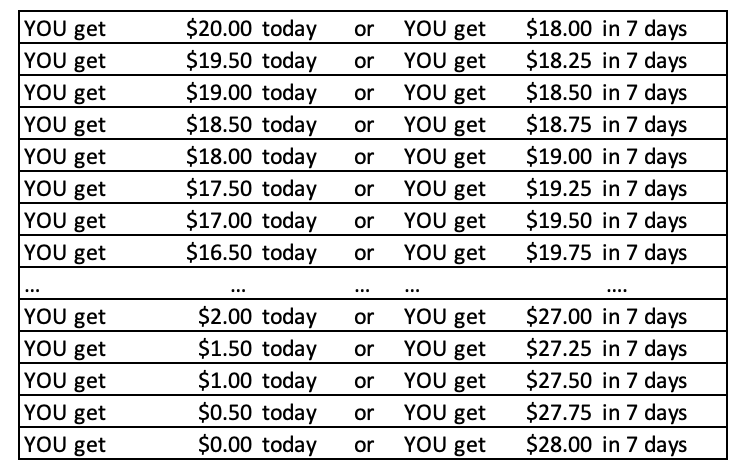}
    \end{center}
    
    {\footnotesize \underline{Notes:} The left panel presents Q2 question for risk preferences elicitation with $p=0.50$. The right panel presents question Q1 for time preferences with $t = 7$ days.}
\end{figure}

\paragraph{Parameters.} The five questions described in Section \ref{sec:Identification} were administered under two parameterizations in each domain. For risk preferences, we used $p =0.50$ and $p= 0.80$. For time preferences, we used $t=7$ and $t=14$ days. In consequence, we obtain $5\times 2 + 5\times 2=20$ rounds. The remaining parameters $(q,K,a,c)$ were set to $(\frac{3}{2},18,\frac{1}{2},28)$ and $(2,18,\frac{1}{2},31)$ for risk questions with $p=0.50$ and $p=0.80$, respectively and to $(\frac{3}{2},20,\frac{1}{2},36)$ and $(2,20,\frac{1}{2},30)$ for time questions with $t=7$ and $t=14$ days, respectively. In the risk questions, the steps between MPL rows were \$1, while in the time questions, it was \$0.50. These parameters were calibrated to capture the relevant preferences regions as established in the vast previous experimental literature that studies time and risk preferences and to minimize the number of dominated choices in each MPL. Table \ref{tab:parameters} in the appendix summarizes all the questions that appeared in the experiment.\footnote{Our experimental design is static in the sense that all decisions are made at a single date. It would be interesting to study behavior in a dynamic choice setting, as in \cite{Halevy2015TimeConsistency}.}

\paragraph{Structure of the Experiment.}The rounds were grouped into four blocks: block A included 4 modified MPL \me-\me questions (Q1); block B included 8 standard MPL \other–\me questions (Q2 and Q3); block C included 4 standard MPL \other–\other questions (Q4); and block D included 4 modified MPL \me–\other questions (Q5). To control for order effects, we used a quasi-randomized block order, subject to the constraint that block A always preceded blocks C and D, since the latter blocks relied on values elicited in block A. After every four questions, we implemented brain breaks to reduce fatigue. During the brain breaks, participants were asked to find the distinction between two pictures depicting animals.

\paragraph{Control tasks.} At the end of the experiment, participants completed a demographic questionnaire stating their age, gender, education, income, number of children and siblings, and political orientation on the liberal-conservative spectrum. Participants also answered two additional questions: one on their financial risk-taking behavior (``What percentage of your retirement savings is invested in risky assets such as stocks or mutual funds?'') and one on prosocial behavior (``How often do you give money to charities?''). 

\paragraph{Incentives.} Participants received a \$4 completion payment. In addition, 20\% of participants were randomly selected for a bonus payment: one row from one of their MPL decisions was randomly implemented for payment. \other recipients were recruited in separate sessions, and bonus payments for choices affecting them were paid accordingly. The study lasted on average 25 minutes. The average total earnings per participant were approximately \$9. This average includes payment for both the \me and the  \other participant, depending on which question was chosen for payment, and the choice made by the participant in this question.

\paragraph{Sample and Recruitment.} The experiment was programmed and administered via Prolific in March 2025, with recruitment restricted to U.S.-based adults. A total of 300 participants completed the study. We excluded 33 participants who made two or more dominated choices—indicative of inattentiveness—leaving a final analysis sample of 267 participants.\footnote{In addition, the completion time for inattentive subjects is significantly shorter than for attentive subjects ($p<0.01$). Inattentive subjects also tend to click only once on the first row of the MPL, allowing them to move quickly through the questions. }

\paragraph{Ethics and Pre-registration.} The study was approved by the Caltech (IRB protocol IR21-1167) and Berkeley Institutional Review Boards (CPHS protocol 2023-05-16357), and was pre-registered on AsPredicted (\#130714).

\subsection{Discussion of Experimental Design}\label{sec:dissdesign}

Two features of our methodology are worth discussing. First, while we describe the approach commonly used in the literature as a no-SIG method, we acknowledge that if the participant (\me) cares about the \other's consumption, then this approach does in fact provide incentives: the participant is choosing between objects that will be consumed by the  \other. However, if the participant is purely selfish and derives no utility from the \other's consumption, then no such incentives exist. By contrast, our SIG approach avoids this problem entirely. It features built-in incentives regardless of whether the participant values the \other's consumption because objects consumed by the  \other are always evaluated in terms of the participant’s own consumption. We view this as a key advantage of our method.

Second, a potential concern with our design is that participants’ answers in Q1 determine the payments used in Q2 and Q3. This raises the possibility that subjects might strategically manipulate their responses in Q1 to influence later choices, thereby threatening incentive compatibility. We believe that this concern is largely mitigated in our setting for several reasons. First, it is highly unlikely that participants could detect the link between Q1 and Q2/Q3; much less learn how to manipulate the interface to their advantage. Participants faced multiple parameterizations of the core questions in randomized order, and the indifference points $x^{\ast}$ elicited in Q1 were transformed into amounts in Q2 and Q3 through algebraic manipulations with constants $(q,a,c,K)$, effectively obscuring the connection between the different questions. Second, as we show in the results section, the recovered distributions of risk and time preferences for \me are similar to what has been found in the literature. This suggests that the answer to Q1 were not manipulated. Third, under the standard ``isolation'' assumption---that participants treat each question as independent---incentive compatibility is preserved. (For empirical support of the isolation assumption, see \cite{KahnemanTversky1979} and \cite{HeyLeea,HeyLeeb}.) For a detailed theoretical discussion of the assumptions underlying incentive compatibility in different payment mechanisms, see \cite{AzrieliEtAl}. Most relevant to our design is their result on the random problem selection mechanism: paying subjects for one randomly selected question is incentive compatible provided preferences satisfy a (state-wise) monotonicity condition.\footnote{If the choice objects are lotteries, then the relevant preferences are over compound lotteries defined by the payment mechanism.}

Our paper is not the first to employ a dynamic design. For further discussion of gaming concerns, see two recent papers that implement alternative versions of dynamic designs: the dynamically optimized sequential experimentation (DOSE) method by \cite{ChapmanEtAl} and the Bayesian rapid optimal adaptive design (BROAD) by \cite{RayEtAl}. Both papers provide evidence that gaming behavior is quite rare and unlikely to undermine preference elicitation.\footnote{There exists an alternative method that can, in principle, eliminate gaming concerns. Applied to our setting, this would require eliciting Q2 and Q3 for all possible values of $x^{\ast}$ from Q1. This would, however, require an infeasible number of questions. \cite{JohnsonEtAl} suggest a workaround for this problem by defining the grand set of such questions, but administering only a subset in the experiment. For payment, one question is randomly drawn from the grand set. If the selected question was part of the administered subset, payment is based on the subject’s original response; if not, the question is posed as an additional one. We do not expect our results to change under this alternative procedure, and therefore leave this as a robustness check for future work.}

\section{Results}\label{sec:results}

This section is organized as follows. We begin with a broad overview of our main findings. The subsequent subsections present the analyses deriving these results. Section \ref{sec:selfish} describes the identification of selfish types. Section \ref{sec:SIGnonpar} presents a non-parametric analysis based on the SIG method, while Section \ref{sec:comparenonpar} compares non-parametric results obtained using the SIG and no-SIG methods. Section \ref{sec:structural} introduces a structural parametric model of preferences and estimates its parameters under both methods. Finally, Section \ref{sec:predict} evaluates the relative performance of the two approaches through an out-of-sample prediction exercise using the estimates from Section \ref{sec:structural}.

\paragraph{Heterogeneity and non-parametric comparisons.}

The application of the SIG method reveals substantial individual heterogeneity in how decision-makers distinguish between their own preferences and those they hold for others. The population is  divided into three groups: those who are more risk averse with respect to \other's than over \me's consumption constitute the largest group; those who are less risk averse are the smallest; while those who are equally risk averse comprise about 30\% of the participants. Despite this heterogeneity, at the level of individual participants, the findings are consistent across tasks and lottery specifications. In the time domain, while the results are more symmetric than in the risk domain, there is a tendency toward greater impatience regarding the other's delays compared to one's own. Crucially, the SIG methodology successfully identifies selfish types, which constitute approximately 20\% of participants in the risk domain and 12\% in the time domain. 

\paragraph{Methodological differences: SIG vs. no-SIG.}

Comparing the SIG results with the traditional no-SIG method exposes clear contradictions between the two approaches. The no-SIG method yields a distribution of types that differs significantly from the SIG method, even when selfish types are excluded. Most notably, in the risk domain, the no-SIG method suggests that the majority of participants are more risk-averse for themselves than for others, directly contradicting the modal finding of the SIG method. Furthermore, the no-SIG approach drastically underestimates the fraction of participants who hold identical preferences for self and other ($<13$\% vs.\ $>30$\% in SIG). This discrepancy suggests that the no-SIG method's reliance on valuations expressed entirely in the ``other's'' domain fails to capture the actual underlying preference structure, likely due to the decoupling of the decision-maker's own utility from the valuation task.

Interestingly, it is not the presence of selfish agents that explains the differences between the SIG and no-SIG findings. Even altruistic participants behave differently in the SIG than in the no-SIG treatments.

\paragraph{Structural Estimation and Predictive Power.}

Our structural estimation results, assuming Constant Relative Risk Aversion (CRRA) and exponential discounting, corroborates the ordinal findings of the SIG method. The estimated risk parameter for \me ($\theta_M=0.62$  is significantly higher than for \other ($\theta_O=0.55$), confirming that agents behave as if they are more risk-averse  when evaluating the other's consumption. Similarly, the estimated discount factor is higher for self ($\delta_M \approx 0.98$) than for the other ($\delta_O \approx 0.95$), indicating greater patience for one's own delays. In contrast, the no-SIG structural estimates are noisy and fail to reject the null hypothesis that preferences are identical across domains. 

\paragraph{Out of sample prediction.}

Finally, an out-of-sample prediction exercise using a holdout question (Q5) demonstrates that the SIG parameter estimates achieve a significantly lower root mean squared error (RMSE) than the no-SIG estimates, validating the superior predictive power and goodness of fit of the skin-in-the-game approach.

\subsection{Identifying selfish types}\label{sec:selfish} An important aspect of our design is that it allows us to categorize some participants as selfish. In each domain, there are six questions in which participants evaluate objects received by others in terms of the objects received by themselves; these are the questions Q2, Q3, and Q4 with different parameterizations. We use responses to these questions to identify \textit{selfish types}: participants who consistently allocate all available resources to their own consumption and prefer even the smallest amount of their own consumption over any amount allocated to others. Such participants seem to derive utility exclusively from their own consumption and none from the consumption of others.

Specifically, we classify a participant as \textit{selfish} in domain $d$ if they choose an option with any positive potential payoff for themselves (even if probabilistic) over an option that provides a positive payoff to the  \other in at least two out of six questions. In other words, such participants select options benefiting the  \other only when their own payoff is strictly zero. Based on this criterion, 54 out of 267 participants (20\%) are selfish in the risk domain, 31 out of 267 (12\%) are selfish in the time domain, and 57 out of 267 (21\%) are selfish in at least one domain. In Appendix \ref{app:additional}, we show that although the proportion of selfish types naturally depends on the chosen cutoff, our qualitative results remain unchanged under both more stringent and more lenient definitions of selfishness.\footnote{Table \ref{tab:SelfishNbQs} presents the distribution of selfish choices in different number of questions. Table \ref{tab:SIGordinalSelfishDef} replicates the main ordinal classification presented in Table \ref{tab:SIGordinal} in Section \ref{sec:SIGnonpar}.}

\subsection{SIG, non-parametric analysis}\label{sec:SIGnonpar}
Table~\ref{tab:SIGordinal} compares own and \other's  risk and time preferences using the ordering of $y^{\ast}$ and $z^{\ast}$ (Section~\ref{sec:Identification}). The data exhibit substantial heterogeneity: some participants are more risk-averse or more impatient for \me than for \other, some are less so, and many display similar attitudes for \me and \other. The distribution of comparative attitudes is stable across parameterizations in each domain. In the risk domain, type classifications for \me vs.\ \other are positively and significantly correlated across parameterizations ($p<0.01$ for all three comparisons). Similarly, in the time domain, classifications are positively and significantly correlated ($p<0.05$ for all three comparisons).
 
\begin{table}[h!]
    \begin{center}
    \caption{Preferences for \me vs \other, ordinal skin-in-the-game approach}\label{tab:SIGordinal}
 {\footnotesize    \begin{tabular}{l|cc|l|cc } 
       & \multicolumn{2}{c|}{Risk domain} & & \multicolumn{2}{c}{Time domain}  \\ 
       \me $--$ than \other  & $p=0.50$  & $p=0.80$ & \me $--$ than \other & $t=7$ days & $t=14$ days  \\ \hline 
   more risk-averse      & 0.24  & 0.28 &  more impatient  & 0.33& 0.28 \\ 
   equally risk-averse      & 0.31  & 0.33 &   equally impatient  & 0.42 & 0.43 \\ 
   less risk-averse     & 0.45 & 0.39  &  less impatient  & 0.24  & 0.29 \\ \hline 
   nb subjects & $i=211$ & $i=208$  & & $i=227$ & $i=225$\\
    \end{tabular}}
\end{center}
    {\footnotesize \underline{Notes:} We focus on non-selfish participants in each domain with non-dominant choice in Q1 in each parameterization. The comparison of attitudes is done using questions Q2 and Q3: if $y^{\ast}<z^{\ast}$ then \me is more risk-averse (more impatient) than \other;  see Section \ref{sec:Identification} for details.}
\end{table}

Our conclusions hold for the strict definition used in Table~\ref{tab:SIGordinal} and the relaxed definition reported in Table~\ref{tab:SIGordinalRelaxed} (Appendix~\ref{app:additional}).\footnote{The relaxed definition allows small perturbations. In the risk domain, we classify \me as more risk-averse than \other if $y^{\ast} < z^{\ast} - \$1$ (the step size in the MPL for risk questions). In the time domain, we classify \me as more impatient than \other if $y^{\ast} < z^{\ast} - \$0.50$ (the step size in the MPL for time questions).}

In the risk domain, more participants are \emph{less} risk-averse for \me than for \other than have the opposite preferences. So the share of those who are \emph{less} risk-averse for \me than for \other exceeds the share of those who are \emph{more} risk-averse for \me than for \other (Test of Proportions: $p<0.05$ in both parameterizations). By contrast, in the time domain, we find no consistent asymmetry across parameterizations: most participants display similar time preferences for \me and \other.\footnote{When $t=7$ days we observe more participants who are more impatient for \me than for \other ($p<0.05$); when $t=14$ days the difference is not significant ($p=0.92$).}

\begin{figure}[h!]
    \begin{center}
    \caption{Connection between preferences, ordinal skin-in-the-game approach}
    \label{fig:YminusZ}
    \includegraphics[scale=0.25]{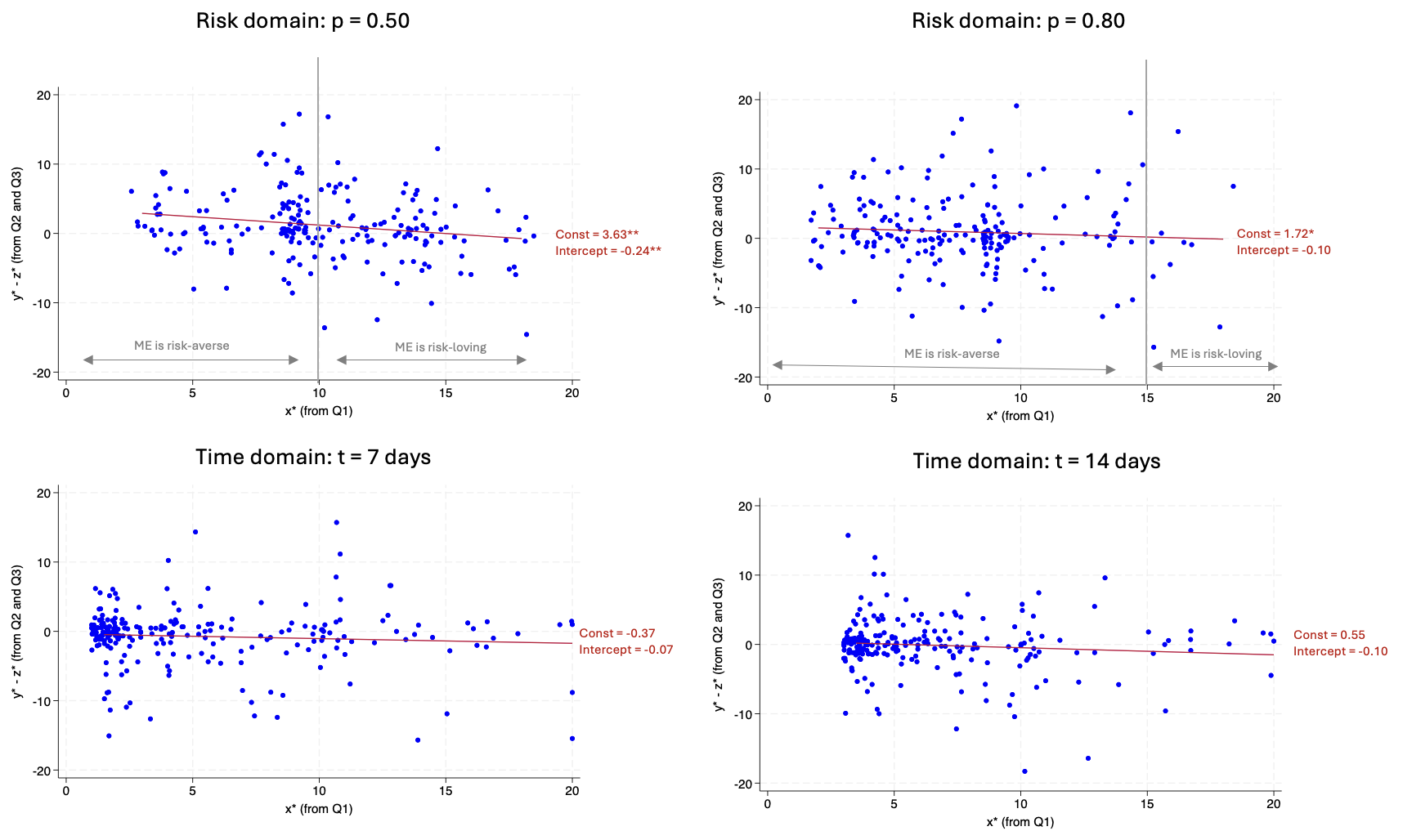}
    \end{center}

    {\footnotesize \underline{Notes:} For each parameterization and each domain, we plot the difference between indifference points $y^{\ast}$ and $z^{\ast}$ versus indifference point $x^{\ast}$. For the risk domain, we split the region of $x^{\ast}$ into the risk-loving and risk-averse/neutral regions. The red line is the linear fit and the estimated intercepts and slopes are recorded next to the line. $^{\ast\ast}$ indicates significance at 5\% level. }
\end{figure}

Figure \ref{fig:YminusZ} shows the relationship between \textsc{me'}s preferences for objects consumed by \other and those consumed by \me, reported separately for each domain. First, we note that for the objects consumed by \me, we reproduce common findings regarding risk and time preferences. In the risk domain, most subjects are risk-averse (75\% for $p=0.50$ and 77\% for $p=0.80$), with a median certainty equivalent of \$9 for a lottery paying \$18.50 with probability 0.50 and \$10 for a lottery paying \$19.50 with probability 0.80. In the time domain, the median subject is indifferent between \$19.75 in 7 days and \$16.50 today (implying a daily interest rate of 2.81\%), and between \$17.50 in 14 days and \$15 today (implying a daily interest rate of 2.38\%). Figure \ref{fig:xstar} in Appendix \ref{app:additional} presents the distribution of $x^{\ast}$ in each parameterization.

The novelty in our design is the comparison between risk and time attitudes for \me and \other. Our data exhibit substantial heterogeneity at the participant level. There is no systematic relationship across participants between own risk preferences and risk preferences for the consumption of \other. The data points in Figure~\ref{fig:YminusZ} are nearly symmetric around zero, indicating that the certainty equivalent of \me consuming a given lottery does not predict \textsc{me'}s certainty equivalent of \other consuming the same lottery (expressed in \me-payments).\footnote{For $p=0.50$, both the intercept and slope of the linear fit are statistically different from zero but their magnitudes are negligible.} An analogous pattern holds for time preferences: an individual’s patience for consumption bundles consumed by \me is unrelated to their patience for bundles consumed by \other. 

Finally, Table \ref{tab:covariates} in the Appendix \ref{app:additional} presents the correlation between the additional control tasks elicited at the end of the experiment and the risk- and time-attitudes we recover for participants for \me vs \other. We find no robust evidence that any of these covariates are related to the differential risk- or time-preferences that people exhibit for their own consumption versus the consumption of others.

\subsection{Comparing SIG and no-SIG, non-parametric analysis}\label{sec:comparenonpar} We next compare the SIG and the no-SIG methods. Table \ref{tab:Compare} reports the classification of types under each method, separately for each domain. Recall that the no-SIG method cannot identify purely selfish types, since the objects consumed by the \other are evaluated only in terms of \textsc{other-}payments. This explains the difference in the number of observations reported in the first two columns of each parameterization. The third column shows the classification based on the no-SIG method after removing the selfish types identified by the SIG method.

The distributions of types recovered by the two methods differ substantially. Consider first the risk domain. Under the no-SIG approach, similar risk attitudes of \me and \other are quite rare---less than 13\% for $p=0.50$ and less than 8\% for $p=0.80$. In contrast, the SIG approach finds that more than 30\% of participants exhibit similar risk attitudes for objects consumed by \me and by \other.\footnote{The differences in the fractions of participants who have similar risk-attitudes in the SIG and no-SIG approaches are statistically significant, $p<0.01$.} Moreover, while the no-SIG approach suggests that the majority of participants are more risk-averse for \me than for \other, the SIG approach identifies the opposite pattern as modal behavior.\footnote{Under the no-SIG approach, the fraction of subjects classified as more risk-averse for \me than for \other is significantly larger than the fraction classified as less risk-averse: $p<0.01$ in both parameterizations.} Interestingly, removing the selfish types and reapplying the no-SIG approach leaves the distribution of types virtually unchanged. This occurs because in the no-SIG case, the majority of purely selfish types are misclassified as being more risk-averse for \me than \other when asked to state the certainty equivalent of the lottery consumed by \other in \other-terms. Indeed, the no-SIG method recovers a similar distribution of risk attitudes for \me and \other when applied only to selfish types (identified by the SIG approach), all non-selfish types, or all subjects together.\footnote{The no-SIG approach applied to selfish participants yields the following type distributions: for $p=0.50$, 50\% are more risk-averse, 17\% equally risk-avers, and 33\% are less risk-averse for \me than \other; for $p=0.80$, 57\% are more risk-averse, 6\% equally risk-averse, and 37\% are less risk-averse for \me than \other.} 

\begin{table}[h!]
    \centering
    \caption{Comparing types recovered by SIG and no-SIG methods}
    \label{tab:Compare}
 {\footnotesize   \begin{tabular}{l|ccc|ccc}
         & \multicolumn{3}{c|}{Risk domain: $p=0.50$}    & \multicolumn{3}{c}{Risk domain: $p=0.80$}\\ \hline 
         & SIG & no-SIG & no-SIG  &   SIG & no-SIG & no-SIG  \\ 
         & & & non-selfish & & & non-selfish \\\hline 
         \me more risk-averse than \other & 0.24& 0.60 & 0.63& 0.28 & 0.56& 0.56\\
         \me equally risk-averse as \other & 0.31& 0.13& 0.12&0.33 & 0.08& 0.08\\
         \me less risk-averse than \other & 0.45& 0.27& 0.25&0.39 & 0.36 & 0.36\\ \hline
         nb subjects & $i=211$& $i=265$ & $i=211$& $i=208$ & $i=262$ & $i=208$ \\
     \hline  \\
        & \multicolumn{3}{c|}{Time domain: $t=7$  days}    & \multicolumn{3}{c}{Time domain: $t=14$ days}\\ \hline 
        & SIG & no-SIG & no-SIG  &   SIG & no-SIG & no-SIG  \\ 
         & & & non-selfish & & & non-selfish \\\hline 
         \me more impatient than \other & 0.33 & 0.44 & 0.45 & 0.28 &0.48  & 0.48 \\
         \me equally impatient as \other  & 0.42 & 0.22 & 0.20& 0.43 & 0.26 & 0.24\\
         \me less impatient than \other  & 0.24 & 0.34 & 0.34& 0.29 & 0.26 & 0.28\\ \hline
         nb subjects & $i=227$& $i=257$ & $i=227$& $i=225$ & $i=254$ & $i=225$ \\
     \hline      
    \end{tabular}}

\vspace{2mm}
    {\footnotesize \underline{Notes:} SIG stands for skin-in-the game approach, while no-SIG stands for no-skin-in-the-game approach.}
\end{table}

In the time domain, the patterns are similar. The fraction of participants who display similar time preferences for \me and \other is significantly smaller under the no-SIG than under the SIG approach.\footnote{We obtain $p<0.01$ in both parameterizations.} Moreover, while the SIG approach recovers roughly equal fractions of participants who are more rather than less impatient for \me than for \other, the no-SIG approach produces a skewed distribution, with substantially more participants classified as more impatient for \me than for \other.\footnote{In the SIG approach, the fraction of participants classified as more impatient for \me than for \other exceeds the fraction classified as less impatient at $t=7$ days ($p=0.03$), but not at $t=14$ days ($p=0.92$). In contrast, under the no-SIG approach, the fraction classified as more impatient is significantly higher than the fraction classified as less impatient in both parameterizations ($p=0.02$ for $t=7$ days; $p<0.01$ for $t=14$ days).} Furthermore, as in the risk domain, removing the selfish types and reapplying the no-SIG approach does not change the distribution of recovered types.\footnote{The no-SIG approach applied to selfish participants yields the following type distributions: for $t=7$ days, 33\% are more impatient, 37\% equally impatient, and 30\% less impatient for \me than \other; for $t=14$ days, 41\% are more impatient, 41\% equally impatient, and 17\% less impatient.}

The similarity between the no-SIG distributions with and without selfish types, combined with the stark differences from the distributions recovered by the SIG approach, suggests that the discrepancies between the two approaches are not merely due to the no-SIG approach’s inability to identify selfish types. Rather, they reflect a more fundamental divergence in how individuals evaluate objects in terms of their own versus others’ consumption.

\subsection{SIG, structural parametric estimation}\label{sec:structural}

We now turn to a structural estimation of a parametric model of preferences. In particular, we assume that participants' utilities over their own consumption, as well as the utility over \other's consumption, takes the form of exponential discounting  with constant relative risk aversion (CRRA). Under these parametric assumptions, we structurally estimate the parameters governing risk and time preferences: the coefficient of relative risk aversion and the discount factor. We use non-linear least squares to jointly estimate four parameters: risk for \me ($\theta_M$), risk for \other ($\theta_O$), time for \me ($\delta_M$), and time for \other ($\delta_O$). This estimation is performed both at the aggregate level (Table \ref{tab:AggEST}) and at the individual level (Figure \ref{fig:IndEstimates}).

The results of the aggregate estimation are reported in Table~\ref{tab:AggEST}. Regression (1) reports the estimated parameters for non-selfish participants in both domains using the three core questions of SIG method. For comparison, regression (2) reports estimates based on the two questions required to recover risk and time parameters under the no-SIG method.

\begin{table}[h!]
    \begin{center}
    \caption{Estimating risk and time preferences parameters for \me and \other}
    \label{tab:AggEST}
   {\footnotesize \begin{tabular}{l|c|c}
   & SIG method & no-SIG method \\
         &  Reg (1) & Reg (2)  \\
  \hline 
     risk parameter for \me   $\theta_M$ & 0.62 (0.02) & 0.63 (0.02)\\
     risk parameter for \other    $\theta_O$  & 0.55 (0.02)& 0.67 (0.04)\\
     time parameter for \me   $\delta_M$ &  0.98 (0.002) & 0.98 (0.002)\\
     time parameter for \other    $\delta_O$  & 0.95 (0.002)  & 0.98 (0.002)\\ \hline 
     H0: $\theta_M = \theta_O$ & $p<0.01$& $p=0.270$ \\
     H0: $\delta_M = \delta_O$ &$p<0.01$ & $p=0.127$ \\ \hline 
          sample   & Q1, Q2, Q3 & Q1, Q4\\  
          nb obs & $n=2412$& $n=2043$\\
          nb subjects & $i=210$ & $i=266$\\
    \end{tabular}}
\vspace{2mm}
\end{center}
    {\footnotesize \underline{Notes:} Nonlinear least-squares estimation is performed assuming CRRA utility and exponential discounting, i.e., $v(x,t) = \delta_M^t \cdot x^{\theta_M}$ and $w(x,t) = \delta_O^t \cdot x^{\theta_O}$. Robust standard errors are reported in the parenthesis with clustering at the participant level.}
\end{table}

Table \ref{tab:AggEST} conveys several interesting findings. First, both the SIG and no-SIG method recover standard preferences for both \me and \other: on average, \me is risk-averse and impatient for both \me and \other, and the parameter estimates are consistent with the range of empirical estimates in the literature. Second, according to the SIG method, \me is more risk-averse when evaluating lotteries consumed by \other than by \me, and more impatient when delays are incurred by \other than by \me. The estimated risk and time parameters for \me and \other differ significantly at the 1\% level. These aggregate estimations align closely with the non-parametric individual-level results presented in the previous section. The only difference between the two is the following: while the non-structural approach shows roughly equal proportions of participants who are more or less impatient for \me than for \other, the structural estimation detects a slight asymmetry in time attitudes. This difference, however, is quantitatively very small.

Third, the no-SIG estimates show that \me is slightly less risk-averse for \other than for \me and is equally impatient for both. However, the estimates are quite noisy: in fact, in both domains the estimated parameters for \me and \other are not statistically different ($p>0.10$). These results are inconsistent with the non-structural no-SIG results reported in the previous section. Specifically, the non-structural no-SIG approach shows that more participants are more risk-averse and more impatient for \me than for \other (Table \ref{tab:Compare}) contradicting the structural estimates.

\begin{figure}
\begin{center}
\caption{Individual Structural Estimates of Risk and Time Parameters}\label{fig:IndEstimates}
\includegraphics[scale=0.13]{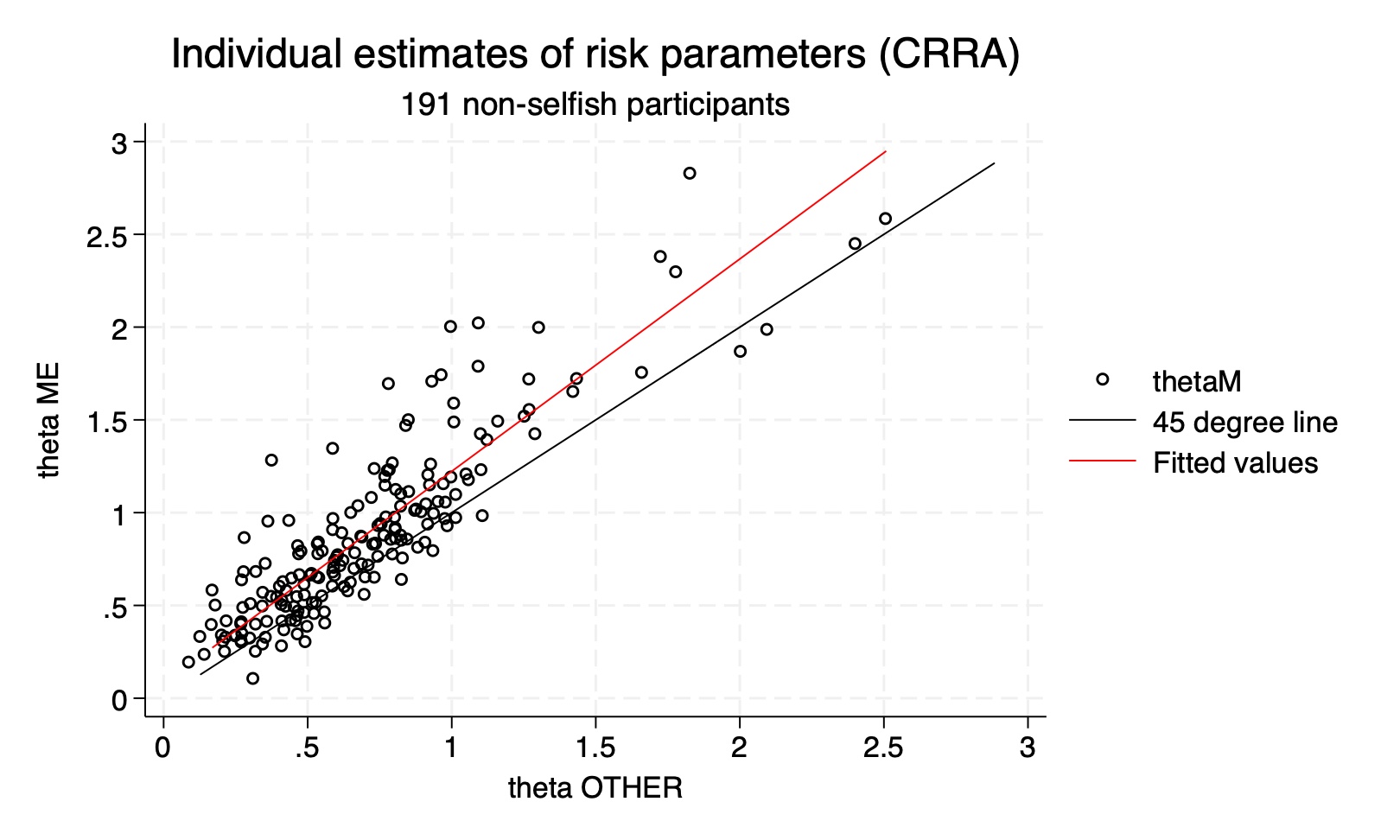} 
\includegraphics[scale=0.13]{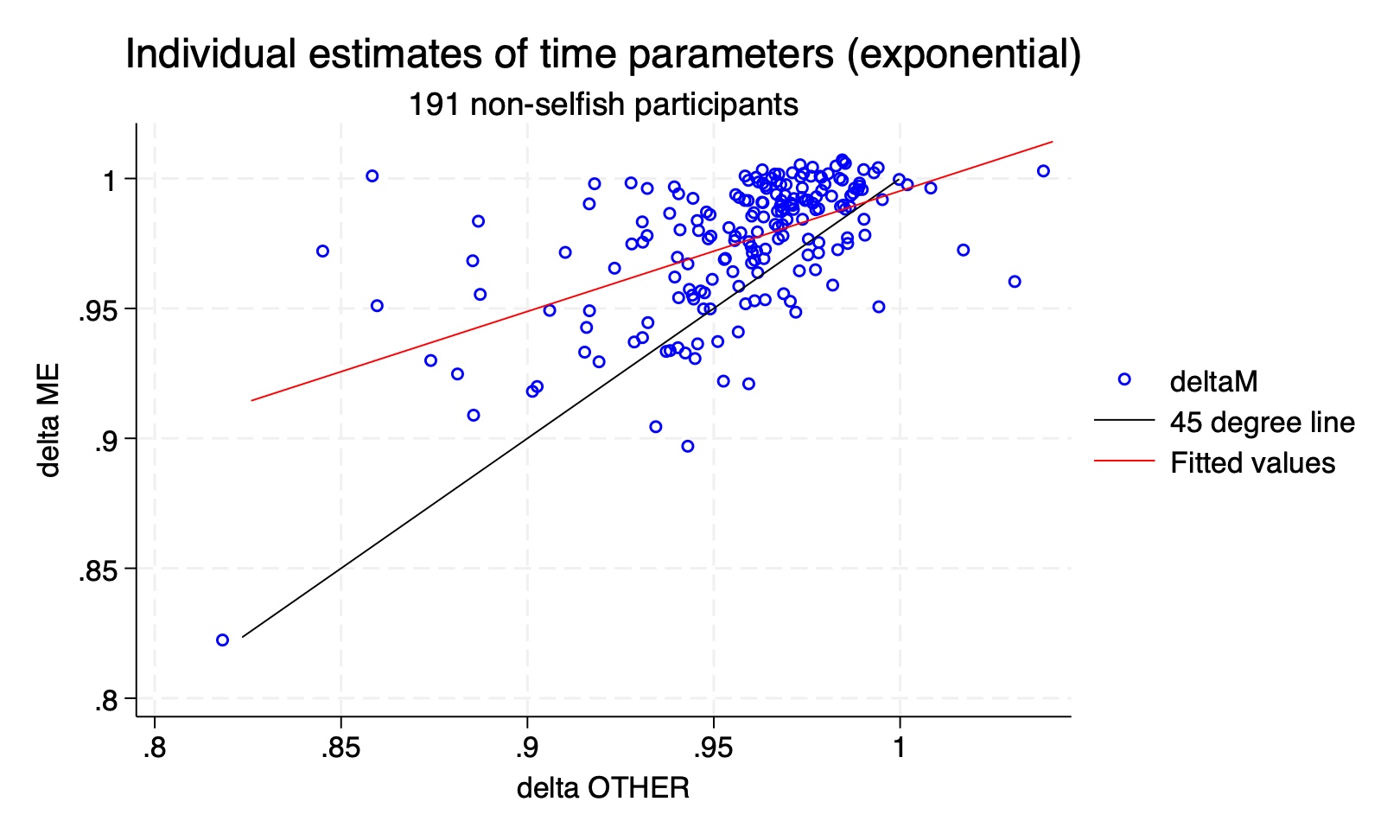}\\
\includegraphics[scale=0.13]{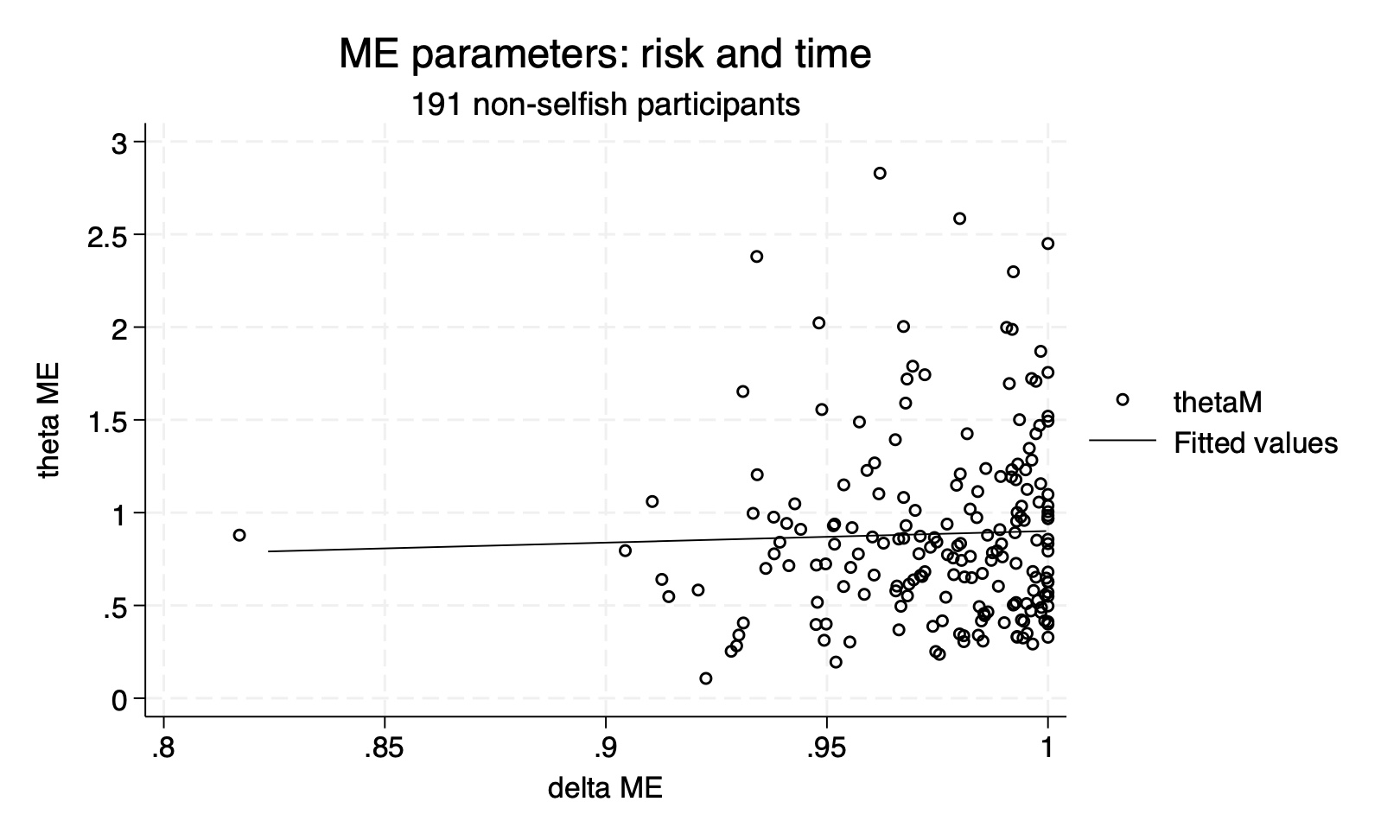} 
\includegraphics[scale=0.13]{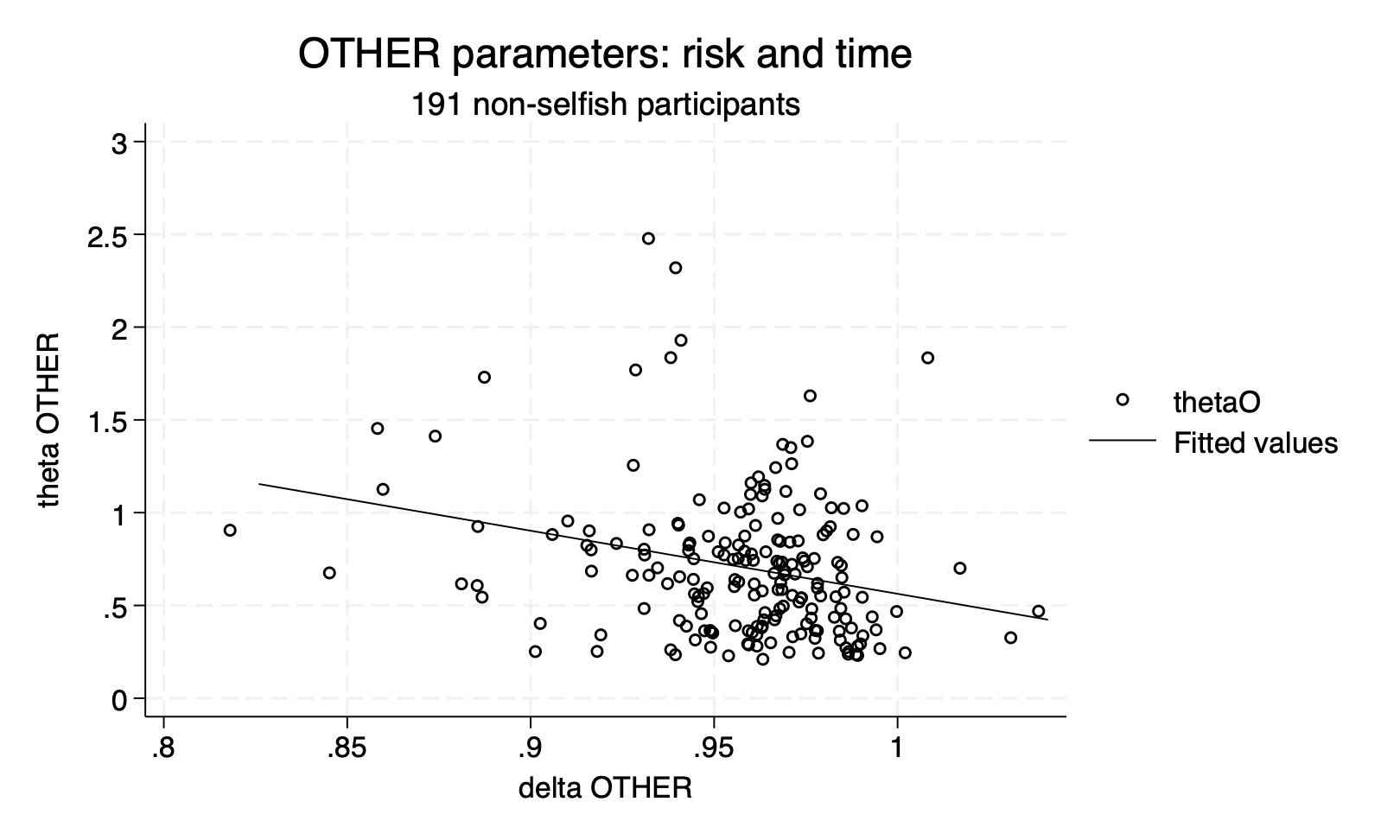}
\end{center}

{\footnotesize \underline{Notes:} Individual estimates of risk and time parameters are plotted for non-selfish participants. We exclude 19 participants for whom we cannot reliably recover their individual estimates. The estimation is based on Q1, Q2, and Q3 in each parameterization in both domains for a maximum of 12 questions per person. The top left (right) panel depicts the relationship between risk (time) preferences for \me (\other). The bottom left (right) panel depicts the relationship between own (other's) parameters in the risk and time domain.  }
\end{figure}

\clearpage
The results of the individual-level estimation are reported in Figure~\ref{fig:IndEstimates}, which graphs our individual-level parameter estimates for the risk and time-preference parameters. The findings in Figure~\ref{fig:IndEstimates} largely corroborate the message of our estimation with aggregate data.  We focus on the 191 non-selfish participants who chose at least five interior certainty equivalents out of the 12 questions used for this estimation. The top two panels compare risk (time) preferences for \me and \other, while the bottom two panels plot both domain parameters separately for \me and \other.

The individual-level SIG estimates closely mirror the aggregate results discussed above. A majority of participants (69\%) are risk-averse for both \me and \other, while 14\% are risk-loving for both. Importantly, majority of participants (81\%) are more risk-averse for \other than for \me, and roughly 10\% display identical risk attitudes across the two. Turning to time preferences, nearly all participants (97\%) are impatient for both \me and \other. Among them, a majority (53\%) exhibit the same degree of impatience, while 41\% are more impatient for \other than for \me.\footnote{To make these individual comparisons we allowed for some slack in the estimates. In particular, we say that \me is more risk-averse for \other than for \me if $\theta_M> \theta_O+0.02$; similarly, we say that \me is more impatient for \other than for \me if $\delta_M> \delta_O+0.02$.}

\subsection{Predictive ability of SIG and no-SIG methods}\label{sec:predict}

In our final exercise, we assess the relative performance of the two methods by conducting an out-of-sample prediction exercise based on the estimated parameters reported in Table~\ref{tab:AggEST}. Standard goodness-of-fit measures such as $R^2$, root MSE, Akaike Information Criterion, or Bayesian Information Criterion are not appropriate in this case, since the SIG and no-SIG methods are estimated on different samples with different numbers of observations. So we use an out of sample prediction exercise to assess both goodness of fit and predictive ability. Specifically, we evaluate how well the parameter estimates obtained under each method predict behavior in Q5, which was not used in the structural estimations. The results indicate that the SIG method achieves a significantly lower root MSE than the no-SIG method (Wilcoxon sign-rank test: $p<0.01$).

We conclude that the structural SIG approach outperforms the no-SIG approach in its out-of-sample predictive ability. Moreover, it produces estimates that are consistent across both the aggregate and individual levels and, importantly, align closely with the non-structural SIG results discussed in the previous section. In contrast, the no-SIG approach yields noisier estimates that lack internal consistency. We view this coherence as an additional strength of the SIG method.

\clearpage

\appendix
\counterwithin{table}{section}
\counterwithin{figure}{section}

\section{Parameters}\label{app:parameters}

Table \ref{tab:parameters} summarizes the parameters used in each question of the experiment. The numbers shown in purple vary depending on the switching row in Question 1 (Q1) for each domain and parameterization:
\begin{itemize}
\item For risk questions with $p=0.50$, we report the Q2–Q4 values corresponding to a decision-maker who is indifferent between receiving \$9 for sure and a lottery paying \$18.50 with probability $0.50$ and \$0 with probability $0.50$.
\item For risk questions with $p=0.80$, we report the Q2–Q4 values corresponding to a decision-maker who is indifferent between receiving \$13 for sure and a lottery paying \$18 with probability $0.80$ and \$0 with probability $0.20$.
\item For time questions with a 7-day delay ($t=7$), we present the Q2–Q4 values corresponding to indifference between receiving \$18.50 today and \$18.75 in 7 days.\item For time questions with a 14-day delay ($t=14$), we present the Q2–Q4 values corresponding to indifference between receiving \$16.50 today and \$16.75 in 14 days.
\end{itemize}

\begin{table}[h!]
    \centering
    \caption{Parameters of experiment}
    \label{tab:parameters}
{\scriptsize     \begin{tabular}{ll|ll|ll|lll|l}
  &  & \multicolumn{2}{l|}{Left option} & \multicolumn{2}{l|}{Right option} & \multicolumn{3}{l|}{Changes across rows} & nb rows \\ \hline 
       &  & recipient & object & recipient & object & $x_{\text{min}}$ & $x_{\text{max}}$ & $x_{\text{step}}$ &  \\ \hline 
         \multicolumn{2}{l|}{\textbf{RISK}} & & & & & & & & \\
         Q1 & $p=0.50$ & \me & \$(18-$x$) & \me & \$(14+0.5$x$) prob $p$ & 0 & 18 & 1 & 19 rows \\
         Q2 & $p=0.50$ & \other & \textcolor{purple}{\$18} & \me & \$$x$ & 0 & 26 & 1 & 27 rows\\
         Q3 & $p=0.50$& \other & \textcolor{purple}{\$37 prob $p$} & \me & \$$x$ & 0 & 26 & 1 & 27 rows\\
         Q4 & $p=0.50$ & \other & \textcolor{purple}{\$18.50 prob $p$} & \other & \$$x$ & 0 & 26 & 1 & 27 rows \\
         Q5 & $p=0.50$ &\me & \$(18-$x$) & \other & \$(28+$x$) prob $p$ & 0 & 18 & 1 & 19 rows \\  \hline 
         Q1 & $p=0.80$ & \me & \$(18-$x$) & \me & \$(15.50+0.5$x$) prob $p$ & 0 & 18 & 1 & 19 rows \\
         Q2 & $p=0.80$ & \other & \textcolor{purple}{\$19.50} & \me & \$$x$ & 0 & 26 & 1 & 27 rows\\
         Q3 & $p=0.80$ & \other & \textcolor{purple}{\$27 prob $p$} & \me & \$$x$ & 0 & 26 & 1 & 27 rows\\
         Q4 & $p=0.80$ &\other & \textcolor{purple}{\$18 prob $p$} & \other & \$$x$ & 0 & 26 & 1 & 27 rows \\
         Q5 & $p=0.80$ &\me & \$(18-$x$) & \other & \$(23.25+0.75$x$) prob $p$ & 0 & 18 & 1 & 19 rows \\  \hline 
         Q1 & $t=7$ days & \me & \$(20-$x$) today & \me & \$(18+$0.5x$) in $t$ & 0 & 20 & 0.50 & 41 rows\\
         Q2 & $t=7$ days & \other & \textcolor{purple}{\$27.75 today} & \me & \$$x$ today & 0 & 20 & 0.50 & 41 rows\\
         Q3 & $t=7$ days & \other & \textcolor{purple}{\$28.13 in $t$} & \me & \$$x$ today & 0 & 20 & 0.50 & 41 rows\\
         Q4 & $t=7$ days & \other & \textcolor{purple}{\$18.75 in $t$} & \other & \$$x$ today & 0 & 20 & 0.50 & 41 rows \\
         Q5 & $t=7$ days & \me & \$(20-$x$) today & \other & \$(27+0.75$x$) in $t$ & 0 & 20 & 0.50 & 41 rows \\ \hline 
         Q1 & $t=14$ days & \me & \$(20-$x$) today & \me & \$(15+$0.5x$) in $t$ & 0 & 20 & 0.50 & 41 rows\\
         Q2 & $t=14$ days & \other & \textcolor{purple}{\$33 today} & \me & \$$x$ today & 0 & 20 & 0.50 & 41 rows\\
         Q3 & $t=14$ days & \other & \textcolor{purple}{\$33.50 in $t$} & \me & \$$x$ today & 0 & 20 & 0.50 & 41 rows\\
         Q4 & $t=14$ days & \other & \textcolor{purple}{\$16.75 in $t$} & \other & \$$x$ today & 0 & 20 & 0.50 & 41 rows \\
         Q5 & $t=14$ days & \me & \$(20-$x$) today & \other & \$(30+$x$) in $t$ & 0 & 20 & 0.50 & 41 rows \\  \hline 
    \end{tabular}}
\end{table}

\clearpage
\section{Additional Analysis}\label{app:additional}

\begin{table}[h!]
    \begin{center}
    \caption{Proportions of subjects behaving selfishly in different number of questions}
    \label{tab:SelfishNbQs}
 {\footnotesize        \begin{tabular}{l|cc}
    \hline \hline 
         & Risk domain & Time domain \\ \hline 
        never &  0.73&0.84 \\
        1 question & 0.07& 0.05 \\
        2 questions & 0.03& 0.02\\
        3 questions &0.02 & 0.005\\
        4 questions &  0.04& 0.04 \\
        5 questions & 0.02&  0.01\\
        6 questions & 0.09 & 0.05 \\ \hline
         nb subjects & 267 &  267 \\
    \hline \hline 
    \end{tabular}}
\end{center}

    {\footnotesize \underline{Notes:} Selfish behavior means choosing options that allocate all available resources to own rather than other's consumption. We consider questions Q2, Q3, and Q4 in each parameterization and in each domain for a total of six questions per domain.}
\end{table}

\begin{table}[h!]
    \begin{center}
    \caption{Risk and time preferences for \me vs \other under different definitions of selfishness, ordinal skin-in-the-game approach}\label{tab:SIGordinalSelfishDef}
 {\footnotesize    \begin{tabular}{l|cc|cc|cc } 
       & \multicolumn{2}{c|}{Main definition} & \multicolumn{2}{c}{Definition 1} & \multicolumn{2}{c}{Definition 2} \\ 
       \me $--$ than \other  & $p=0.50$  & $p=0.80$ & $p=0.50$  & $p=0.80$ & $p=0.50$  & $p=0.80$  \\ \hline 
   more risk-averse      & 0.24  & 0.28 & 0.24 & 0.28 & 0.26&  0.29  \\ 
   equally risk-averse      & 0.31  & 0.33 & 0.33&0.35 & 0.31& 0.32 \\ 
   less risk-averse     & 0.45 & 0.39  & 0.42& 0.37 & 0.44& 0.39 \\ \hline 
   nb subjects & $i=211$ & $i=208$  & $i=192$ & $i=191$  & $i=226$ & $i=223$\\
    \end{tabular}}
    
    \vspace{5mm}

{\footnotesize    \begin{tabular}{l|cc|cc|cc } 
       & \multicolumn{2}{c|}{Main definition} & \multicolumn{2}{c}{Definition 1} & \multicolumn{2}{c}{Definition 2} \\ 
       \me $--$ than \other  & $t=7$ days  & $t=14$ days  & $t=7$ days  & $t=14$ days & $t=7$ days  & $t=14$ days \\ \hline 
   more impatient      & 0.33  & 0.28 &  0.33 & 0.27& 0.33&  0.28 \\ 
   equally impatient     & 0.42  & 0.43& 0.43& 0.44 & 0.42 & 0.43\\ 
   less impatient     & 0.24 & 0.29  & 0.24 & 0.29&0.25 & 0.29 \\ \hline 
   nb subjects & $i=227$ & $i=225$  & $i=215$ & $i=214$ &  $i=232$ &   $i=229$\\
    \end{tabular}}
\end{center}
    {\footnotesize \underline{Notes:} We focus on non-selfish participants in each domain with non-dominant choice in Q1 in each parameterization. The main definition uses the definition of selfishness defined in Section XX in the paper, according to which participant is classified as non-selfish if they choose options that allocate all available resources to own rather than other's consumption in less than 2 out of 6 questions. Definition 1 is more stringent: non-selfish participants never allocate all available resources to only own consumption. Definition 2 is more lenient: non-selfish participant allocate all resources to own consumption in at most 3 out of 6 questions. The comparison of attitudes is done using questions Q2 and Q3: if $y^{\ast}<z^{\ast}$ then \me is more risk-averse (impatient) than \other;  see Section \ref{sec:Identification} for details.}
\end{table}

\begin{table}[h!]
    \begin{center}
    \caption{Preferences for \me vs \other, the `relaxed' definition, ordinal skin-in-the-game approach}\label{tab:SIGordinalRelaxed}
 {\footnotesize    \begin{tabular}{l|cc|l|cc } 
       & \multicolumn{2}{c|}{Risk domain} & & \multicolumn{2}{c}{Time domain}  \\ 
       \me $--$ than \other  & $p=0.50$  & $p=0.80$ & \me $--$ than \other & $t=7$ days & $t=14$ days  \\ \hline 
   more risk-averse      & 0.20 & 0.22 &  more impatient  & 0.29& 0.22 \\ 
   equally risk-averse      & 0.43  & 0.44 &   equally impatient  & 0.52 & 0.53 \\ 
   less risk-averse     & 0.40 & 0.34  &  less impatient  & 0.19  & 0.25 \\ \hline 
   nb subjects & $i=211$ & $i=208$  & & $i=227$ & $i=225$\\
    \end{tabular}}
\end{center}
    {\footnotesize \underline{Notes:} We focus on non-selfish participants in each domain with non-dominant choice in Q1 in each parameterization. The comparison of attitudes is done using questions Q2 and Q3 allowing for small perturbations. If $y^{\ast}<z^{\ast} - \$1$ then \me is more risk-averse than \other, while if $y^{\ast}>z^{\ast} +\$1$ then \me is less risk-averse than \other. Similarly, if   $y^{\ast}<z^{\ast} - \$0.50$ then \me is more impatient than \other, while if $y^{\ast}>z^{\ast} + \$0.50$ then \me is less impatient than \other.}
\end{table}

\begin{figure}[h!]
    \begin{center}
    \caption{Risk and time preferences for \me}
    \label{fig:xstar}
    \includegraphics[scale=0.25]{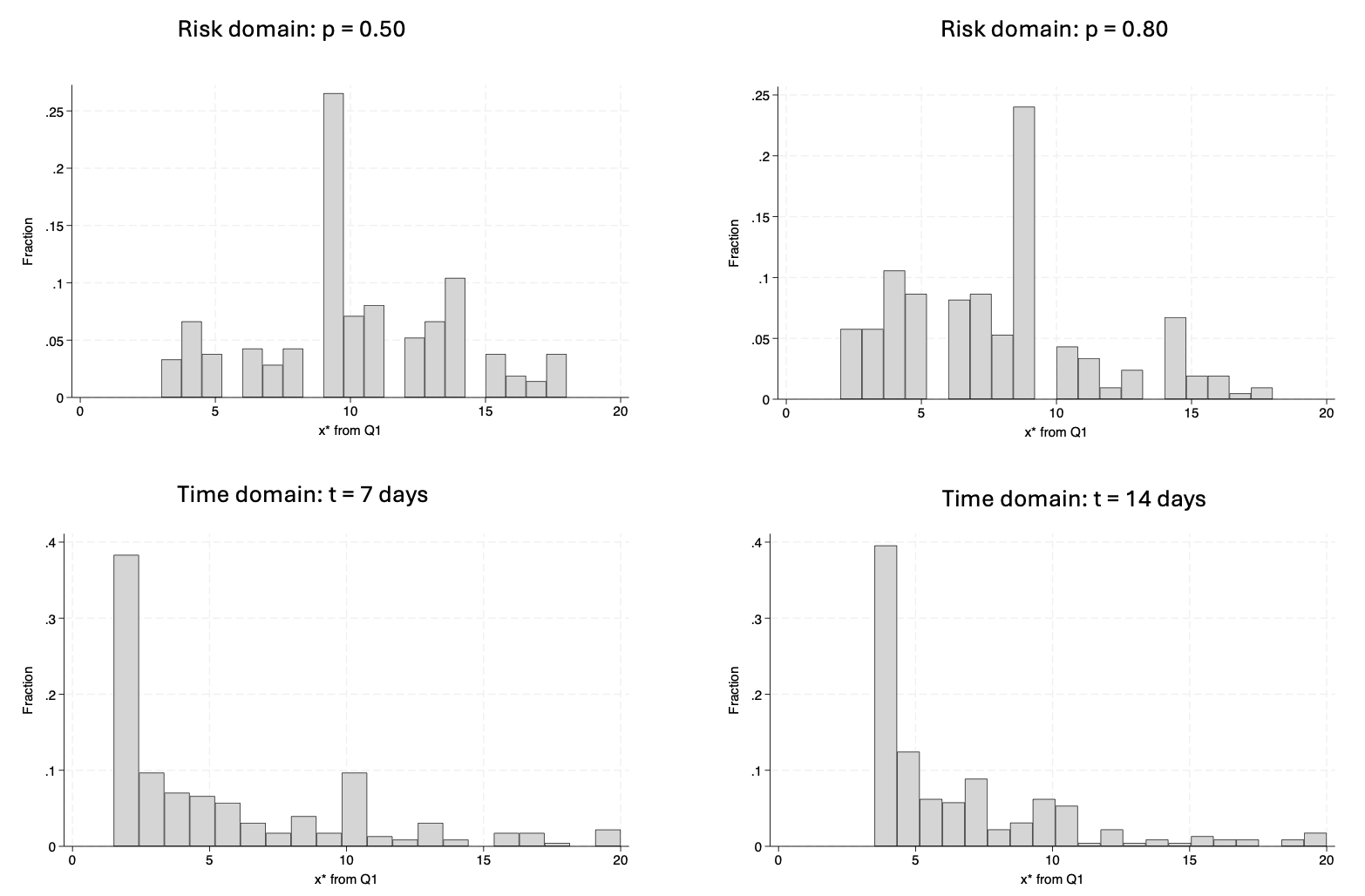}
    \end{center}

    {\footnotesize \underline{Notes:} We plot the histogram of $x^{\ast}$ elicited in Q1 of the corresponding domain/parameterization focusing on non-dominated choices of non-selfish participants.   }
\end{figure}

\begin{table}[h!]
    \begin{center}
    \caption{Effect of Covariates on Risk and Time Preferences for \me vs \other}
    \label{tab:covariates}
   {\scriptsize \begin{tabular}{l|ccc|ccc}
         &  \multicolumn{6}{c|}{Risk domain} \\
         &  \multicolumn{3}{c}{$p=0.50$}  &  \multicolumn{3}{c}{$p=0.80$} \\ 
         & more  & same & less & more & same & less \\  
         &  risk-averse &  &  risk-averse  &  risk-averse &  &  risk-averse \\ \hline 
         Male & -0.04 (0.06) & 0.25$^{\ast \ast}$ (0.06)& -0.20$^{\ast \ast}$ (0.06)& 0.01 (0.06) & 0.04 (0.07) & -0.04 (0.07)\\
         Age (over 35) & -0.002 (0.06)& 0.08 (0.06)& -0.08 (0.07)& -0.01 (0.06)& 0.05 (0.07) & -0.04 (0.07)\\
         College Degree (completed) & -0.14$^{\ast\ast}$ (0.06)& 0.004 (0.07)& 0.14$^{\ast \ast}$ (0.07)& -0.02 (0.07) & -0.09 (0.07) & 0.12$^{\ast}$ (0.07)\\ 
         No children & -0.07 (0.06)& -0.13$^{\ast \ast}$ (0.07) & 0.20$^{\ast \ast}$ (0.07)& -0.003 (0.07) & -0.05 (0.07) & 0.05 (0.07)\\
         No siblings & 0.16 (0.12)& -0.14 (0.14) &-0.02 (0.15) &0.04 (0.15)& -0.05 (0.07) & 0.14 (0.15)\\
         Liberal & 0.10$^{\ast}$ (0.06)& -0.09 (0.06)& -0.007 (0.07)& 0.06 (0.06)& 0.06 (0.06) & -0.12$^{\ast}$ (0.07)\\
         Charity (almost never) & -0.04 (0.10)& 0.16$^{\ast}$ (0.09)& -0.12 (0.12)&  -0.08 (0.11) & 0.04 (0.11) & 0.05 (0.11)\\
         Retirement (not risky) & 0.01 (0.06)& 0.08 (0.06)& -0.09 (0.07) & 0.05 (0.06)  & 0.03 (0.07) & -0.08 (0.07)\\ \hline
         nb of participants & $i=213$ & $i=213$& $i=213$& $i=213$& $i=213$&$i=213$  \\
         
    \end{tabular}}

\vspace{3mm}

       {\scriptsize \begin{tabular}{l|ccc|ccc}
         &  \multicolumn{6}{c|}{Time domain} \\
         &  \multicolumn{3}{c}{$t=7$ days}  &  \multicolumn{3}{c}{$t=14$ days} \\ 
         & more  & same & less & more & same & less \\  
         &  impatient &  &  impatient & impatient &  &  impatient \\ \hline 
         Male & 0.001 (0.06)& 0.09 (0.06)& -0.09$^{\ast}$ (0.06)& 0.03 (0.06)& -0.001 (0.06)&-0.02 (0.06) \\
         Age (over 35) & -0.14$^{\ast \ast}$ (0.06)& 0.21$^{\ast \ast}$ (0.06)& -0.07 (0.05)& -0.12$^{\ast \ast}$ (0.06)&0.15$^{\ast \ast}$ (0.06) &  -0.03 (0.06)\\
         College Degree (completed) & -0.02 (0.07)& -0.10 (0.07)& 0.12$^{\ast}$ (0.06)& 0.004 (0.06)& -0.06 (0.07)& 0.06 (0.07)\\
         No children & -0.02 (0.07)& 0.01 (0.07)& 0.01 (0.06)& -0.01 (0.06)& 0.14$^{\ast \ast}$ (0.07)& -0.12$^{\ast}$ (0.06)\\
         No siblings & 0.05 (0.12)& -0.04 (0.13)& -0.008 (0.11)& 0.19$^{\ast}$ (0.10)& -0.20 (0.14)& 0.003 (0.13)\\
         Liberal & 0.11$^{\ast}$ (0.06)& -0.13$^{\ast\ast}$ (0.06)& 0.03 (0.05) & -0.007 (0.06)& -0.09 (0.06)& 0.10$^{\ast}$ (0.06)\\
         Charity (almost never) & 0.06 (0.10)&-0.21$^{\ast}$ (0.10) & 0.14$^{\ast}$ (0.09)& -0.006 (0.09)&-0.08 (0.11) & 0.09 (0.10) \\
         Retirement (not risky) & 0.09 (0.06)& -0.03 (0.06)& -0.06 (0.06)& 0.09 (0.06)& -0.10 (0.06)& 0.01 (0.06)\\ \hline
         nb of participants & $i=236$ & $i=236$& $i=236$& $i=236$& $i=236$&$i=236$  \\
         
    \end{tabular}}
    \end{center}

    {\footnotesize \underline{Notes:} We estimate Multinomial Logit regression with dependent variable recording whether the subjects is more, same, or less risk-averse for \me than \other in the top panel and more, same, less impatient impatient for \me than for \other in the bottom panel for each parameterization separately. We then report the average marginal effect (AME) of each covariate on the probability that the outcome takes each of the three values. In both domains we focus on non-selfish participants as defined in our results section. }
\end{table}

\clearpage
\section{Experimental instructions}

\input{instructions.tex}

\clearpage
\bibliographystyle{ecta}
\bibliography{references}

\end{document}

%% file: macros.tex
\def\Re{\mathbf{R}}

\def\ta{\theta}

\def\phi{\varphi}

\newcommand{\df}[1]{\textit{#1}}

\def\me{\textsc{Me}\xspace}
\def\other{\textsc{Other}\xspace}

\newtheorem{proposition}{Proposition}
\newtheorem{example}{Example}

\usepackage{graphicx}
\usepackage[cmyk,x11names]{xcolor}

\usepackage{pgf,pgfplots,import}
\pgfplotsset{compat=newest}

\definecolor{accentcolor}{RGB}{50, 70, 150}

\usepackage[font=normalsize,justification=centering]{caption}
\DeclareCaptionLabelFormat{captionf}{\textsc{#1~#2}}
\usepackage{booktabs}
\usepackage{multirow}
\usepackage{subcaption}

\usepackage{lscape} 
\usepackage{placeins}
\providecommand{\U}[1]{\protect\rule{.1in}{.1in}}

\usepackage{setspace}

%% file: instructions.tex
\newcommand{\iprem}{P}

\newcommand{\drawrect}[1]{\tikz \fill [#1] (0,0) rectangle (0.5,0.3);}
\newcommand{\drawcirc}[1]{\tikz \fill [#1] (0,0) circle (2pt);}
\newcommand{\drawlegend}{
	\begin{tabular}{cccc}
    \multicolumn{4}{c}{\bf Legend}\\
    \toprule
    	seeking & neutral & averse &\\
    \midrule
    	\drawrect{ColorRed!80} & \drawrect{ColorBlue!80} & \drawrect{ColorGreen!80} & $\% >0$\\
    	\drawrect{ColorRed!50!white} & \drawrect{ColorBlue!50!white} & \drawrect{ColorGreen!50!white} & $\% =0$\\
    	\drawrect{ColorRed!20!white} & \drawrect{ColorBlue!20!white} & \drawrect{ColorGreen!20!white} & $\% <0$\\
    	\drawcirc{ColorRed!20!white} & \drawcirc{ColorBlue!20!white} & \drawcirc{ColorGreen!20!white} & data\\
    	\bottomrule
    \end{tabular}
}

\subsection*{WELCOME}

\noindent Thank you for participating in this experiment. We ask you to avoid opening other websites and apps until you finish all tasks. The experiment will take approximately 30 minutes.\\

\noindent To begin the experiment, please enter your Prolific ID:

\subsection*{PAYMENT}

\noindent The study consists of 20 choice tasks. You will get \$5 for completing the study. \\

\noindent In addition, we will randomly select one out of every five participants to receive a bonus payment. If you are selected to receive a bonus, then at the end of the study the computer will randomly select one of the tasks. Your bonus payment will depend on your responses in this randomly selected task.\\

\noindent Note two things: 

\begin{enumerate} 
\item There are no right or wrong answers. We are interested in studying \textbf{your preferences.}
\item Any one of the decisions you make today may determine your bonus payment. Therefore, it is in your best interest to \textbf{make every decision carefully and in a way that reflects what you truly prefer.}
\end{enumerate}

\noindent NEXT

\subsection*{CHOICES}

\noindent Each task will involve several decisions between either lotteries or money bundles. Let's talk about these objects first.\\

\noindent NEXT

\subsection*{LOTTERIES}

\noindent A lottery pays different amounts depending on the chance. Here is an example of a lottery:
\begin{center}
50\% chance of \$10\\
50\% chance of \$5
\end{center}

\noindent This lottery pays either \$10, with a 50\% chance, or \$5, with a 50\% chance. For your payment, the computer will randomly draw a number between 1 and 100, where each number is equally likely to be drawn. If the drawn number is less or equal to 50, the lottery will pay \$10. If the drawn number is above 50, the lottery will pay \$5. In other words, the lottery pays either \$5 or \$10 with equal chance.\\

\noindent Depending on the question, the chance could be different: for example, it could be 25\%, 75\%, etc. In some cases, the lottery will involve no chance at all: for example, the option may just pay \$12 for sure.\\

\noindent NEXT\\

\noindent Consider the following lottery:
\begin{center}
    20\% chance of \$7 \\
80\% chance of \$4
\end{center}

\noindent What is the chance that you will receive \$4 if this lottery determines your bonus? (in percentage)\\
 
\noindent What is the chance that you will receive \$7 if this lottery determines your bonus? (in percentage)\\
 
\noindent NEXT\\

\subsection*{MONEY BUNDLES}

\noindent A money bundle pays different amounts of money on different dates. Here is an example of a money bundle:
\begin{center}
\$10 in 7 days
\end{center}

\noindent This bundle pays \$10 in seven days from now. If this bundle is chosen to determine your bonus, then you will receive \$10 seven days after completing the study.\\

\noindent NEXT\\

\noindent Consider the following money bundle:
\begin{center}
\$5 in 12 days
\end{center}

\noindent What will be your bonus payment in 12 days if this bundle determines your bonus?\\

\noindent NEXT

\subsection*{WHAT HAPPENS IN A CHOICE TASK}

\noindent Your decisions in a choice task will be presented in a table in which \textbf{each row represents a separate decision}. In each row, you will choose between the Left and the Right options. Here is an example with 6 rows:

\begin{center}
PRESS TO SEE AN EXAMPLE
\end{center}

\begin{figure}[h!]
    \centering
    \includegraphics[scale=0.5]{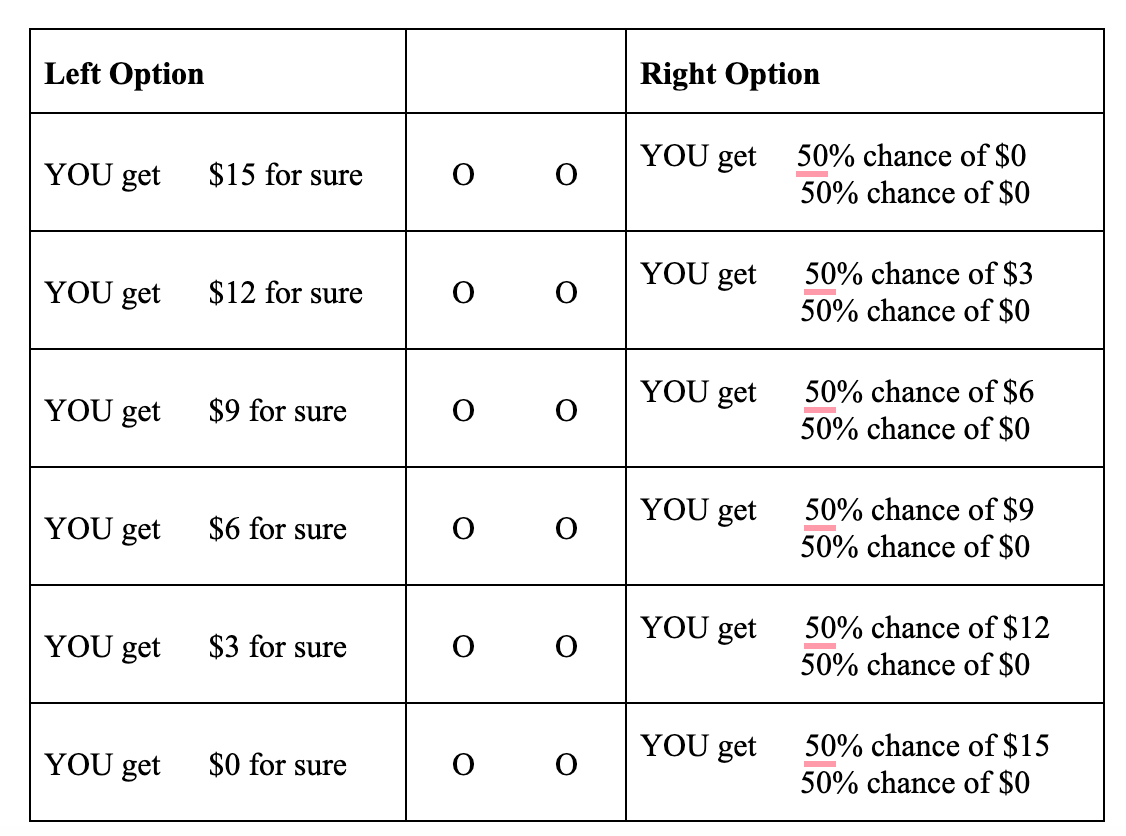}
\end{figure}

\noindent In each row, you must select one of the two options: Left or Right. We will explain on the next screen how you will be making decisions.

\subsection*{HOW TO MAKE DECISIONS}

\noindent To avoid having to make a selection in each row, you need only click once to indicate when one option starts to be preferred to the other. For example, say you click on the Left option in a particular row. Then all rows above that row will populate to indicate that you choose the Left option, while all the rows below the one you clicked will populate to indicate that you choose the Right option. Similarly, if you click on the Right option in a particular row, then all the rows above the one you clicked will indicate that you chose the Left option and all the rows below the one you clicked will indicate that you chose the Right option.\\

\noindent For instance, if you clicked on the Left option in row 3, then it would mean that in rows 1 to 3 you chose the Left option, while in rows 4 to 6 you chose the Right option.\\

\begin{figure}[h!]
    \centering
    \includegraphics[scale=0.4]{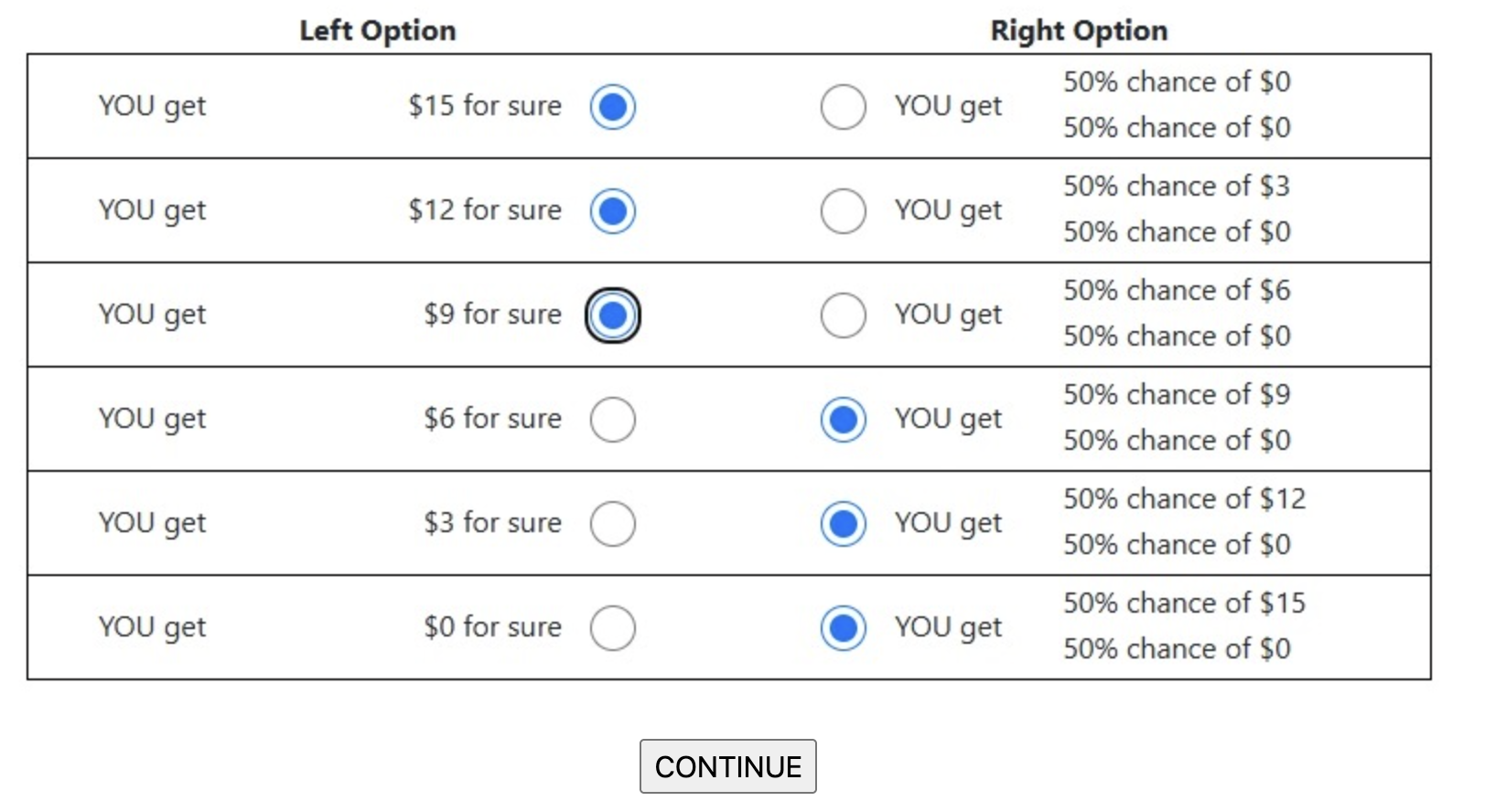}
\end{figure}

\noindent If you want to select the Left option in all rows, then you need to click on the Left option in the very last row.\\

\begin{figure}[h!]
    \centering
    \includegraphics[scale=0.4]{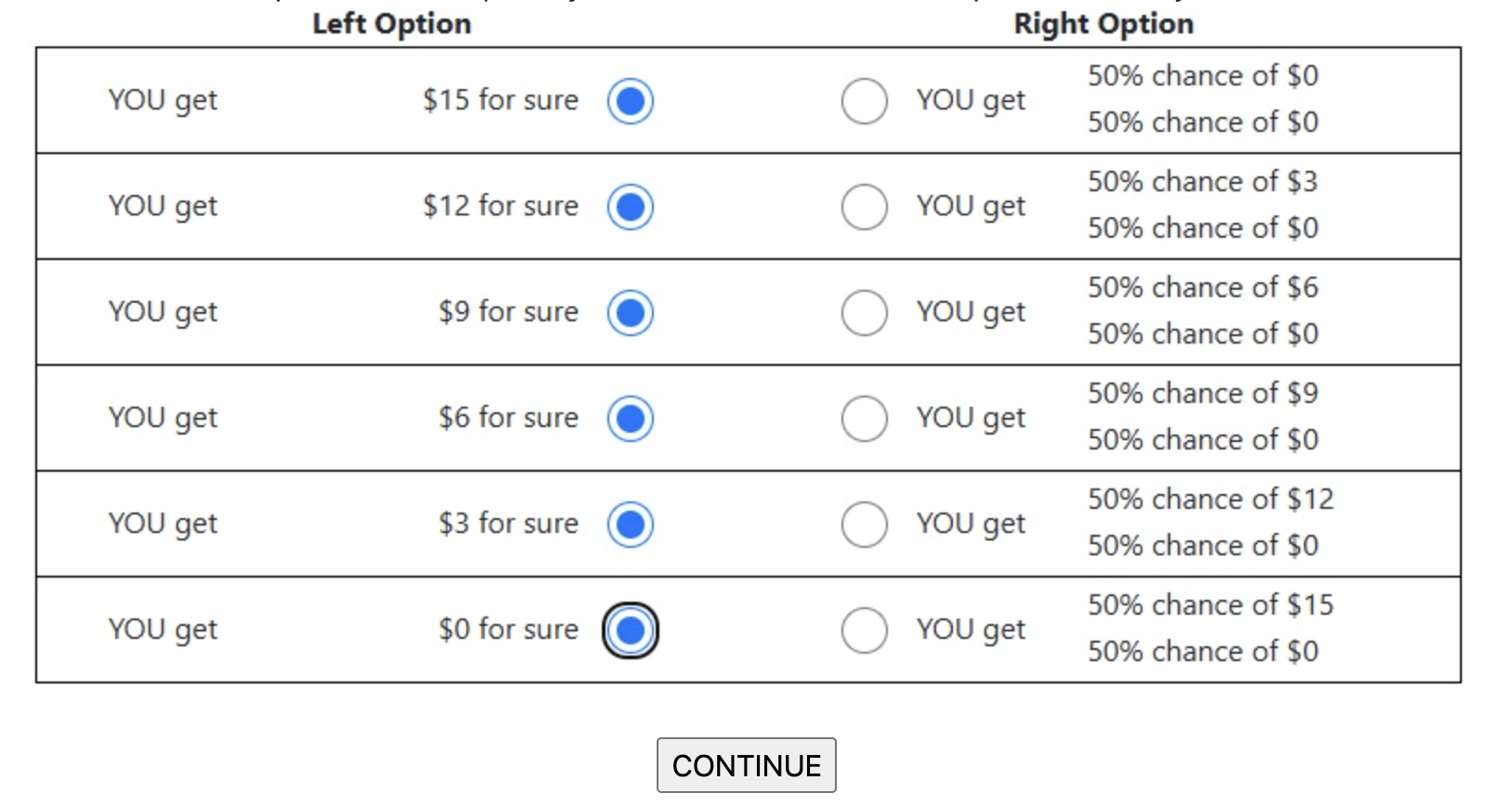}
\end{figure}

\noindent If you want to select the Right option in all rows, then you need to click on the Right option in the very first row.\\

\begin{figure}[h!]
    \centering
    \includegraphics[scale=0.4]{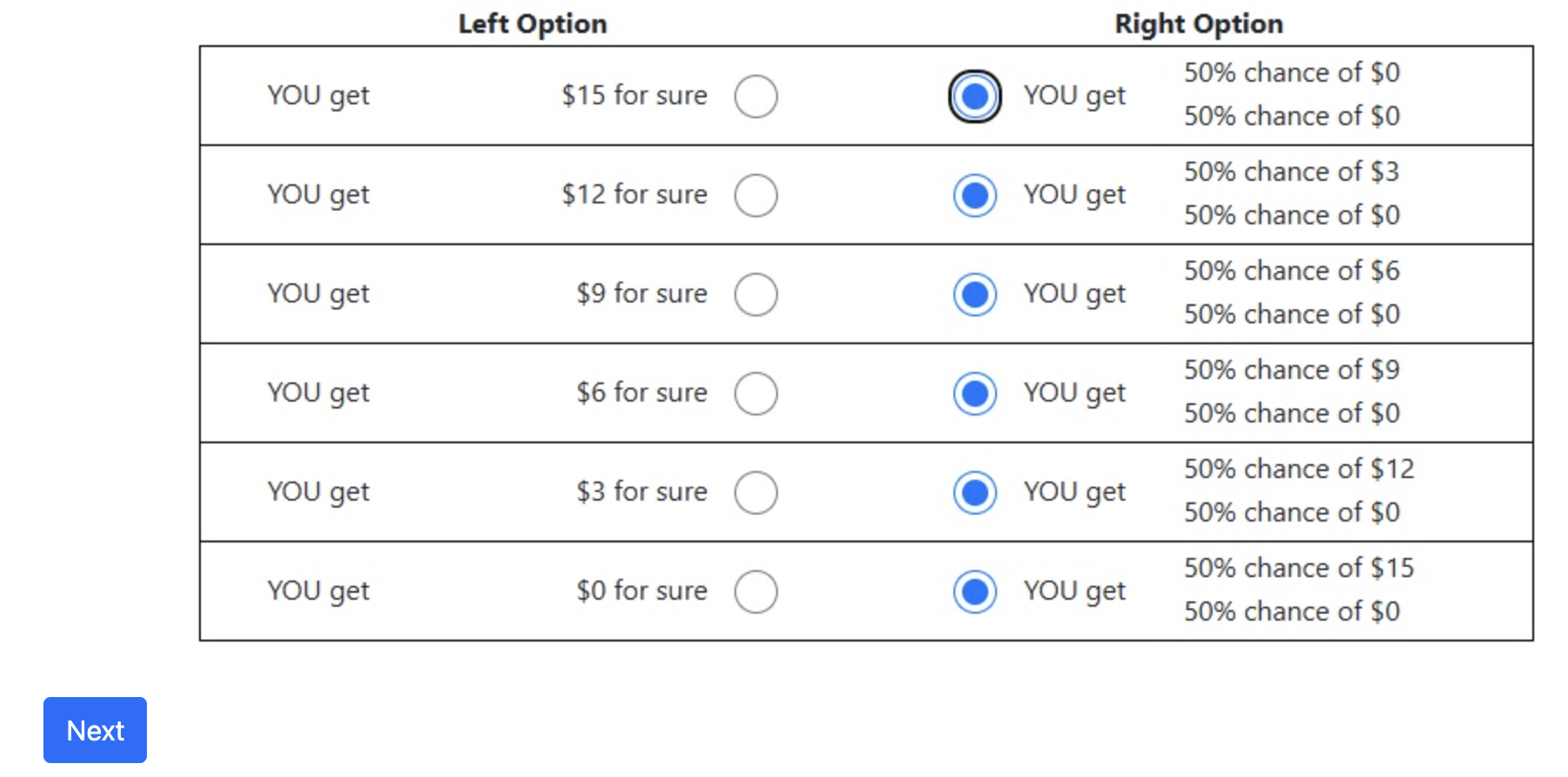}
\end{figure}

\section*{Instructions for Block A}

\subsection*{WHAT HAPPENS IN THE NEXT FOUR ROUNDS}

\noindent In the next four choice tasks, you will face a series of decisions in which you will be choosing between objects YOU might receive. Here is an example of a choice task with 6 rows:

\begin{center}
PRESS TO SEE AN EXAMPLE
\end{center}

\begin{figure}[h!]
    \centering
    \includegraphics[scale=0.5]{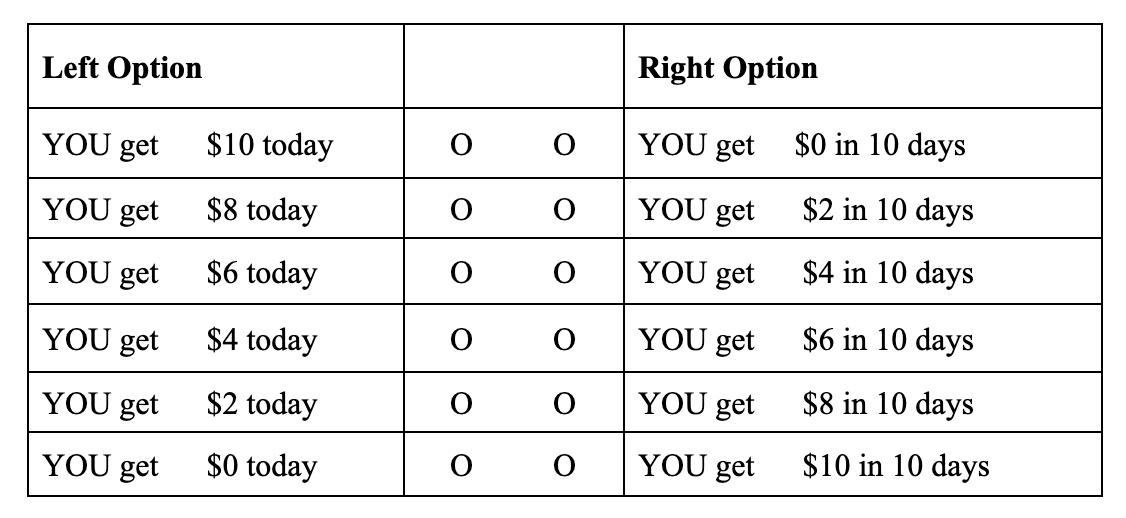}
\end{figure}

\noindent Each row is a separate decision. In each row, you must select one of the two options: Left or Right. Both the Left and the Right options are objects received by YOU. \\

\noindent Note: \textcolor{purple}{The Left option pays less money as we go down the rows}. Contrary to that, \textcolor{blue}{the Right option gets better and pays more money as we go down the rows}.\\

\noindent In some rows, one of the two options pays more than the other. For example, in the first row, the Left option pays \$10 today, while the Right option pays \$0 in 10 days. Another example is the last row: the Left option pays \$0 today, while the Right option pays \$10 in 10 days. \\

\noindent Finally, in some choice tasks, you will be choosing between amounts and lotteries, while in others you will be choosing between money bundles.\\

\noindent NEXT

\subsection*{INTERFACE PRACTICE QUESTION 1:}

\noindent Please select the Left option in the first 5 rows and the Right option in the last row:

\begin{figure}[h!]
    \centering
    \includegraphics[scale=0.5]{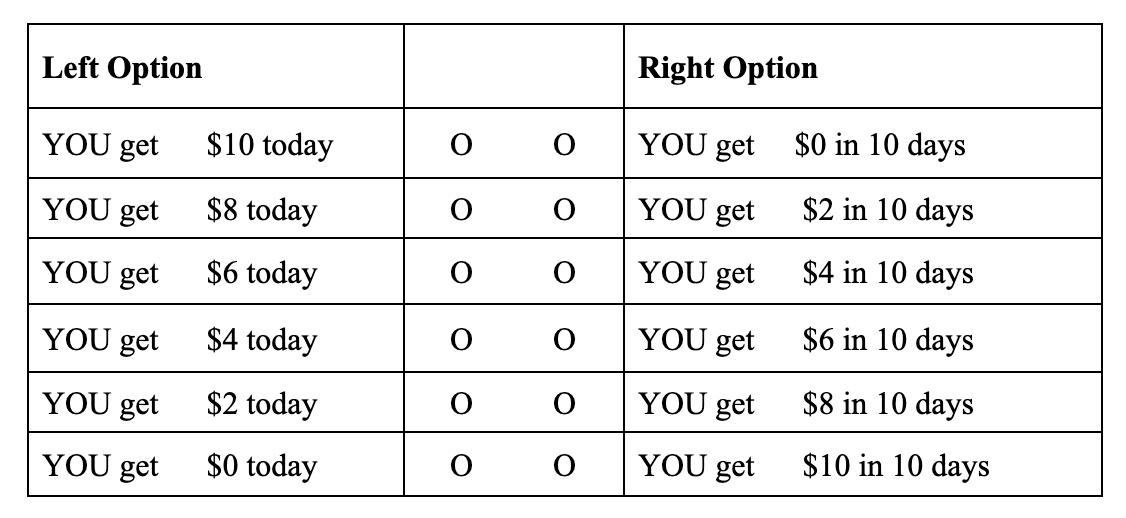}
\end{figure}

\noindent NEXT

\subsection*{INTERFACE PRACTICE QUESTION 2:}

\noindent PPlease select the Right option in all rows:

\begin{figure}[h!]
    \centering
    \includegraphics[scale=0.5]{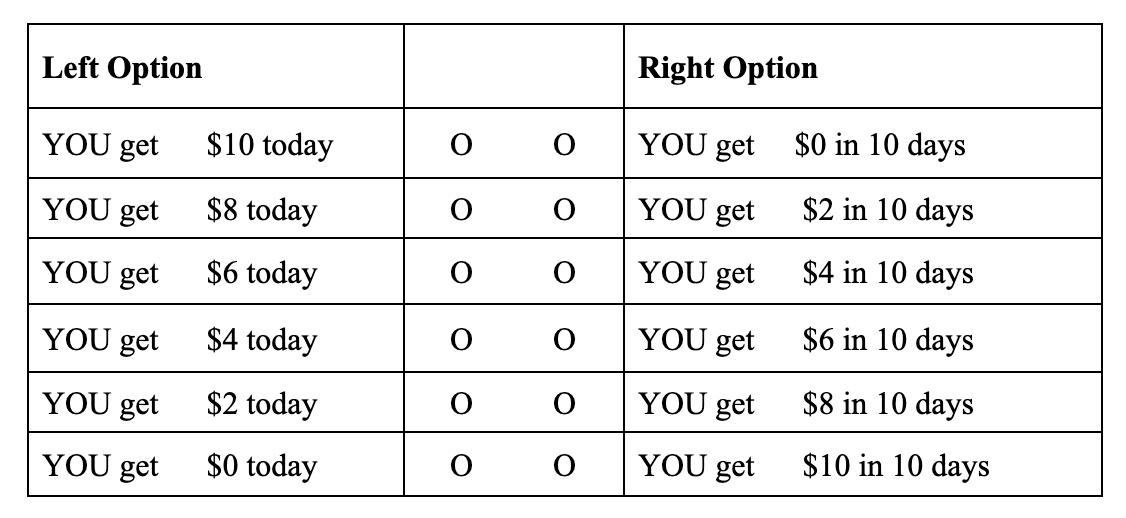}
\end{figure}

\noindent NEXT \\

\noindent Congratulations! You have demonstrated that you understand the next block of rounds.\\

\noindent Please press the NEXT button to continue the study. \textbf{Please make sure you answer every decision carefully and in a way that reflects your true preferences!}

\section*{Instructions for Block B}

\subsection*{WHAT HAPPENS IN THE NEXT EIGHT ROUNDS}

\noindent In the next eight choice tasks, you will be choosing between objects \textcolor{purple}{YOU} might receive and objects that \textcolor{blue}{ANOTHER PROLIFIC WORKER} might receive.\\

\noindent Specifically, after you complete the study, we will recruit another Prolific worker who is similar to you in age and level of education. Let’s call this person \textbf{OTHER}. The OTHER has no active role in our study. They will simply receive the objects that you will choose for them in the next rounds, if one of these rounds is selected for a bonus.\\

\noindent Here is an example of a choice task with 6 rows:
\begin{center}
PRESS TO SEE AN EXAMPLE
\end{center}

\begin{figure}[h!]
    \centering
    \includegraphics[scale=0.5]{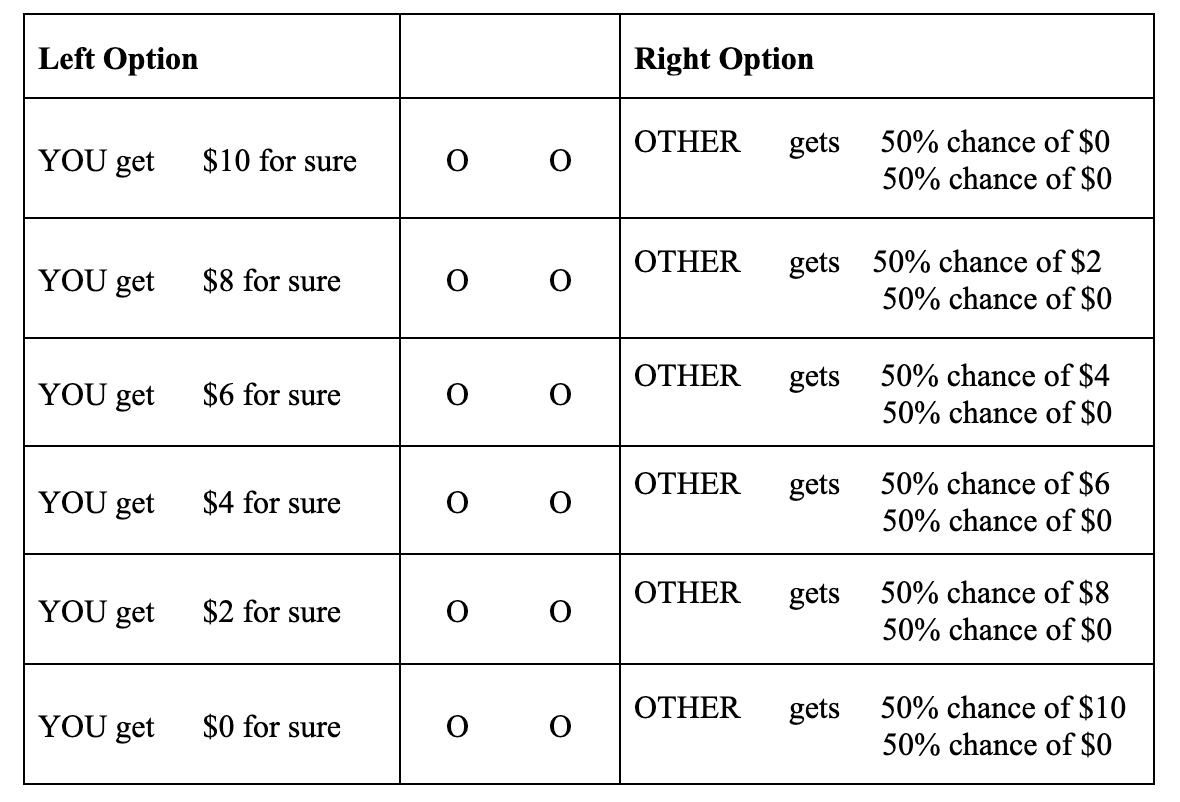}
\end{figure}

\noindent As always, each row is a separate decision. In each row, you choose between Left and Right. In every row (decision), \textcolor{purple}{Left specifies different amounts for YOU}, while \textcolor{blue}{Right is different lotteries for OTHER}. As you go down the rows, amounts for YOU become smaller and smaller (see the Left option), while the lottery for OTHER becomes better and better (see the Right option). \\

\noindent You will be prompted to choose between amounts for YOU and lotteries for OTHER in every single row. For instance, say you choose Left in the first four rows and Right in the 2 remaining rows. Say, also, that row two is selected for a bonus. Then YOU will receive \$8 for sure. If instead the last row is selected for a bonus, then OTHER will receive a lottery which pays \$10 or \$0 with equal chance.\\

\noindent Finally, in some choice tasks, you will be choosing between amounts for YOU and lotteries for OTHER (like in the example above), while in other choice tasks, you will be choosing between amounts for YOU and money bundles for OTHER.\\

\noindent NEXT

\subsection*{INTERFACE PRACTICE QUESTION 1:}

\noindent Please implement the following selection: choose amounts for YOU as long as these amounts are at least \$6 and choose the lottery for OTHER in all other decisions.

\begin{figure}[h!]
    \centering
    \includegraphics[scale=0.5]{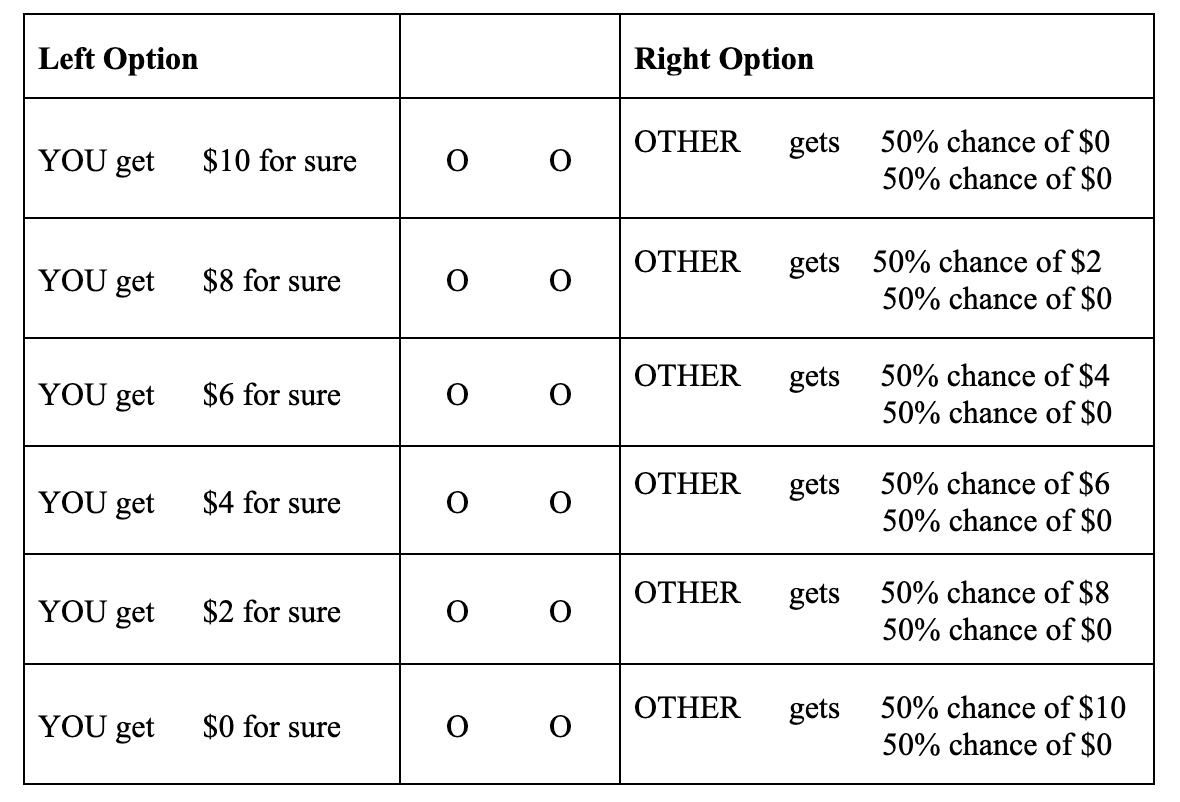}
\end{figure}

\noindent NEXT

\subsection*{INTERFACE PRACTICE QUESTION 2:}

\noindent Please implement the following selection: choose the lottery for OTHER as long as the lottery’s highest payment is no less than \$8 and choose the amounts for YOU in all other decisions.

\begin{figure}[h!]
    \centering
    \includegraphics[scale=0.5]{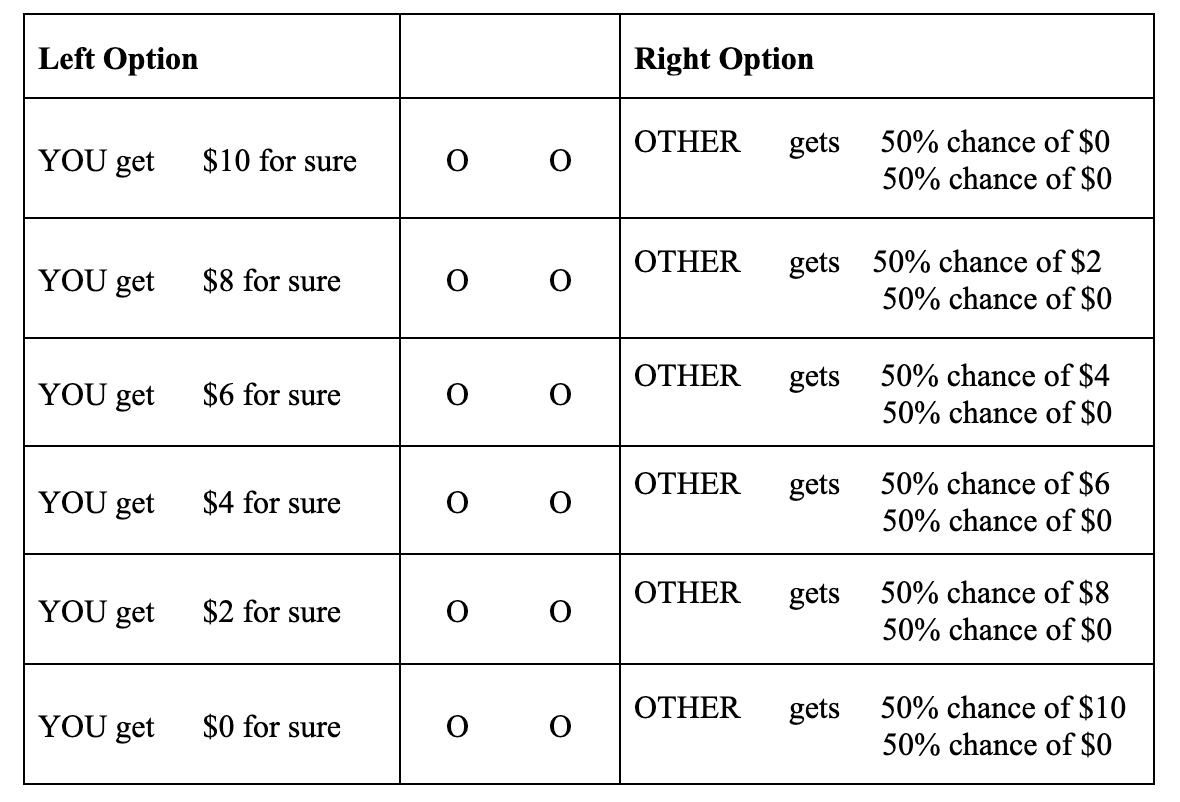}
\end{figure}

\noindent NEXT\\

\noindent Congratulations! You have demonstrated that you understand the next block of rounds.\\

\noindent Please press the NEXT button to continue the study. \textbf{Please make sure you answer every decision carefully and in a way that reflects your true preferences!}

\section*{Instructions for Block C}

\subsection*{WHAT HAPPENS IN THE NEXT EIGHT ROUNDS}

\noindent In the next eight choice tasks, you will be choosing between objects YOU might receive and objects that ANOTHER PROLIFIC WORKER might receive.\\

\noindent Specifically, after you complete the study, we will recruit another Prolific worker who is similar to you in age and level of education. Let’s call this person OTHER. The OTHER has no active role in our study. They will simply receive the objects that you will choose for them in the next rounds, if one of these rounds is selected for a bonus.\\

\noindent Here is an example of a choice task with six rows:

\begin{center}
PRESS TO SEE AN EXAMPLE
\end{center}

\begin{figure}[h!]
    \centering
    \includegraphics[scale=0.5]{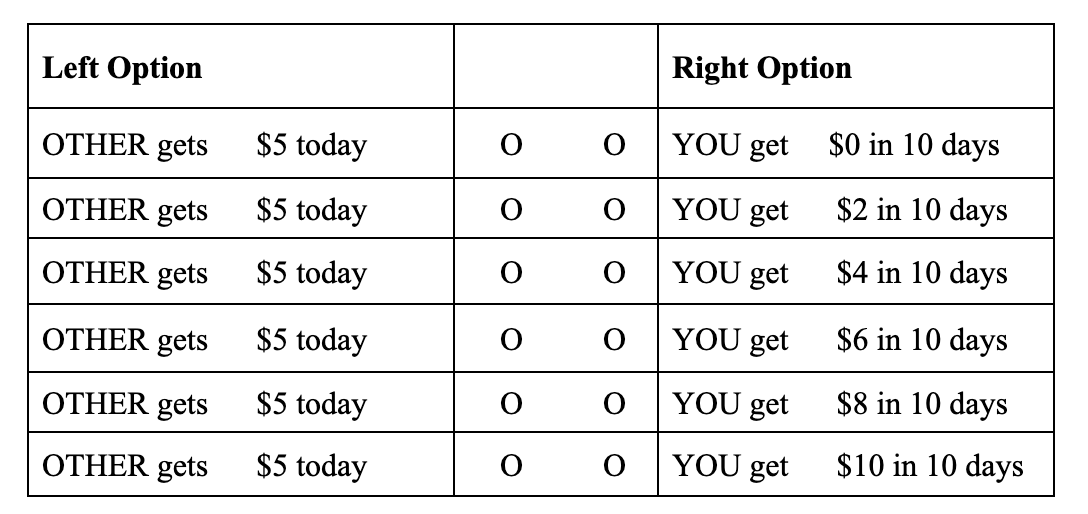}
\end{figure}

\noindent Note: \textcolor{blue}{The objects for OTHER will be displayed on the Left} and \textcolor{purple}{the objects for YOU will be displayed on the Right}. These objects will be either lotteries and sure amounts or money bundles.\\

\noindent Moreover, the Left options, i. e., the objects for OTHER, will be the same in all decisions in all rows of the same choice task, while the Right options, i.e., the objects for YOU, will get better as you go down the rows.\\

\noindent As always, each row is a separate decision. You will be asked to choose either the Left option or the Right option in every single row. \\

\noindent NEXT\\

\subsection*{INTERFACE PRACTICE QUESTION 1:}

\noindent Please implement the following selection: choose the money bundle for OTHER unless the amount for YOU received in 10 days is \$4 or greater.

\begin{figure}[h!]
    \centering
    \includegraphics[scale=0.5]{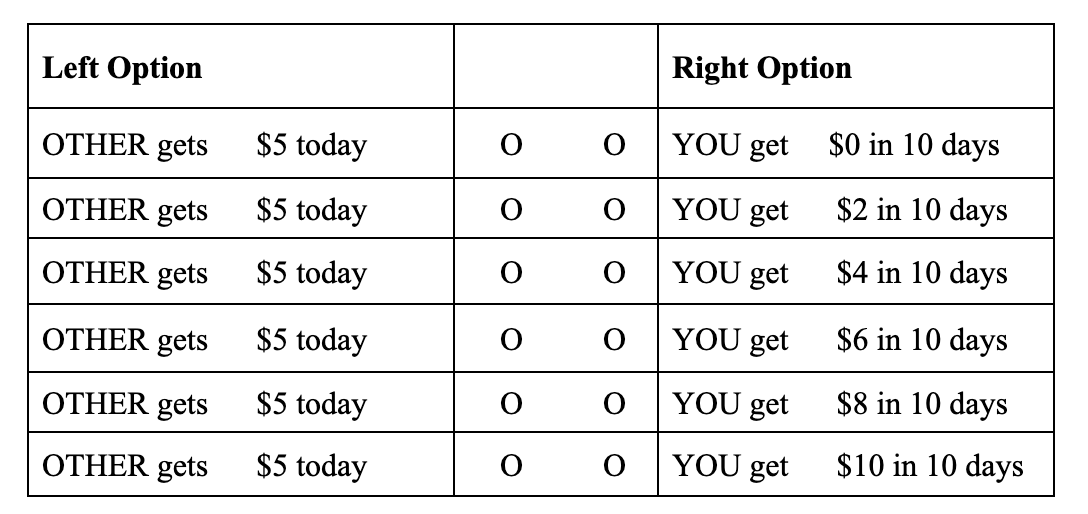}
\end{figure}

\noindent NEXT\\

\clearpage
\subsection*{INTERFACE PRACTICE QUESTION 2:}

\noindent Please implement the following selection: choose the \$5 today for OTHER in all rows.\\

\begin{figure}[h!]
    \centering
    \includegraphics[scale=0.5]{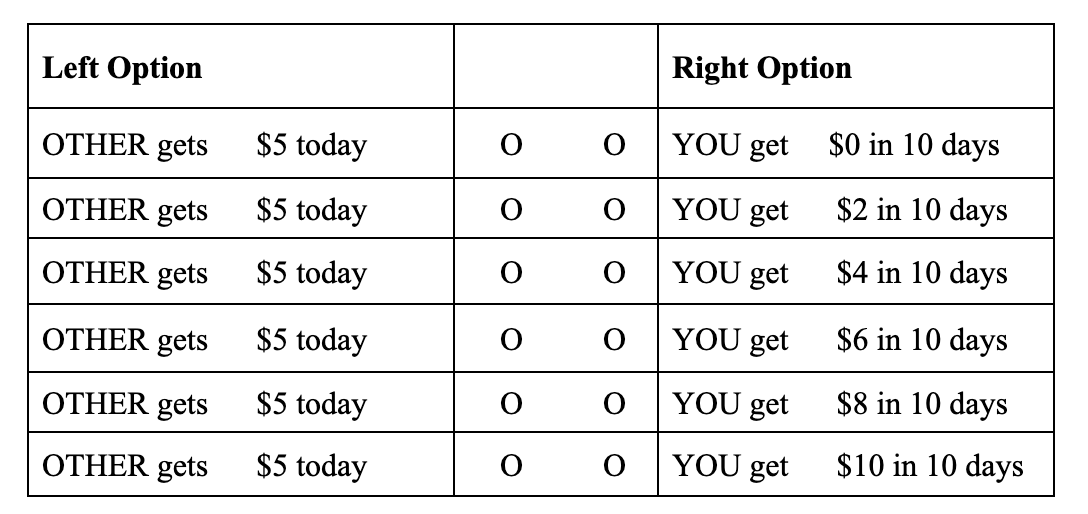}
\end{figure}

\noindent NEXT\\

\noindent Congratulations! You have demonstrated that you understand the next block of rounds.\\

\noindent Please press the NEXT button to continue the study. \textbf{Please make sure you answer every decision carefully and in a way that reflects your true preferences!}

\section*{Instructions for Block D}

\subsection*{WHAT HAPPENS IN THE NEXT FOUR ROUNDS}

\noindent In the next four choice tasks, you will be choosing between objects that ANOTHER PROLIFIC WORKER might receive.\\

\noindent Specifically, after you complete the study, we will recruit another Prolific worker who is similar to you in age and level of education. Let’s call this person OTHER. The OTHER has no active role in our study. They will simply receive the objects that you will choose for them in the next rounds, if one of these rounds is selected for a bonus.\\

\noindent Here is an example of a choice task with six rows:
\begin{center}
PRESS TO SEE AN EXAMPLE
\end{center}

\begin{figure}[h!]
    \centering
    \includegraphics[scale=0.5]{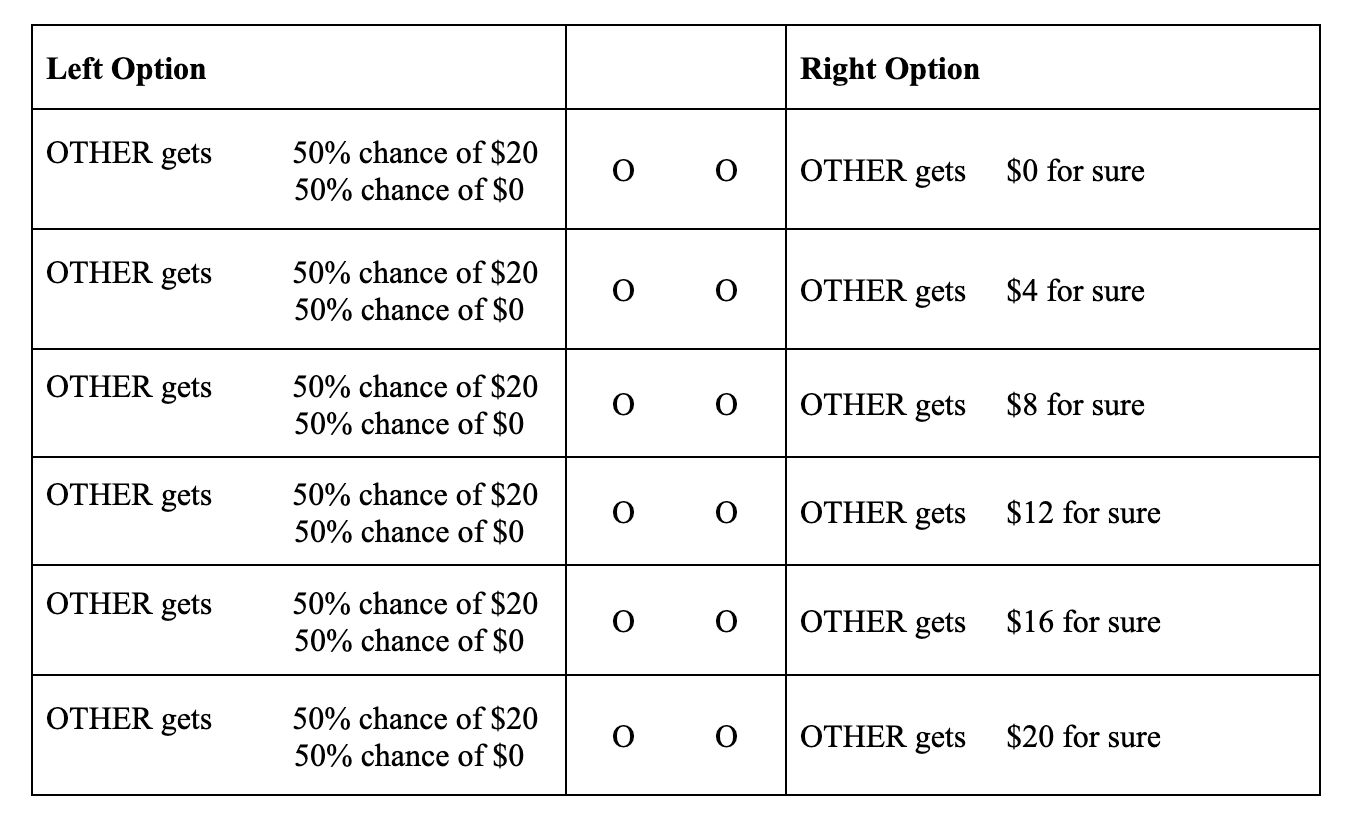}
\end{figure}

\noindent Note: Both, the Left and the Right options list \textcolor{blue}{objects for OTHER}. It will be either lotteries and sure amounts, or money bundles.\\

\noindent Moreover, \textbf{the Left option is the same} in all decisions in all rows of the same choice task, while the \textbf{Right options gets better and better as you go down the rows}.\\

\noindent In some rows, one of the two options pays to the OTHER at least as much and sometimes more than the other. For example, in the first row, the Left option pays \$20 with 50\% chance and \$0 with 50\% chance, while the Right option pays \$0 for sure. Thus, the Right option pays the OTHER at least as much and sometimes more than the Right option. Another example is the last row: the Left option pays \$20 with 50\% chance and \$0 with 50\% chance, while the Right option pays \$20 for sure. Thus, the Right option pays the OTHER at least as much and sometimes more than the Left option.\\

\noindent As always, each row is a separate decision. You will be asked to choose either the Left option or the Right option in every single row. \\

\noindent NEXT\\ 

\clearpage

\subsection*{INTERFACE PRACTICE QUESTION 1:}

\noindent Please implement the following selection: choose the sure amount for the OTHER as long as this amount is at least \$4.

\begin{figure}[h!]
    \centering
    \includegraphics[scale=0.5]{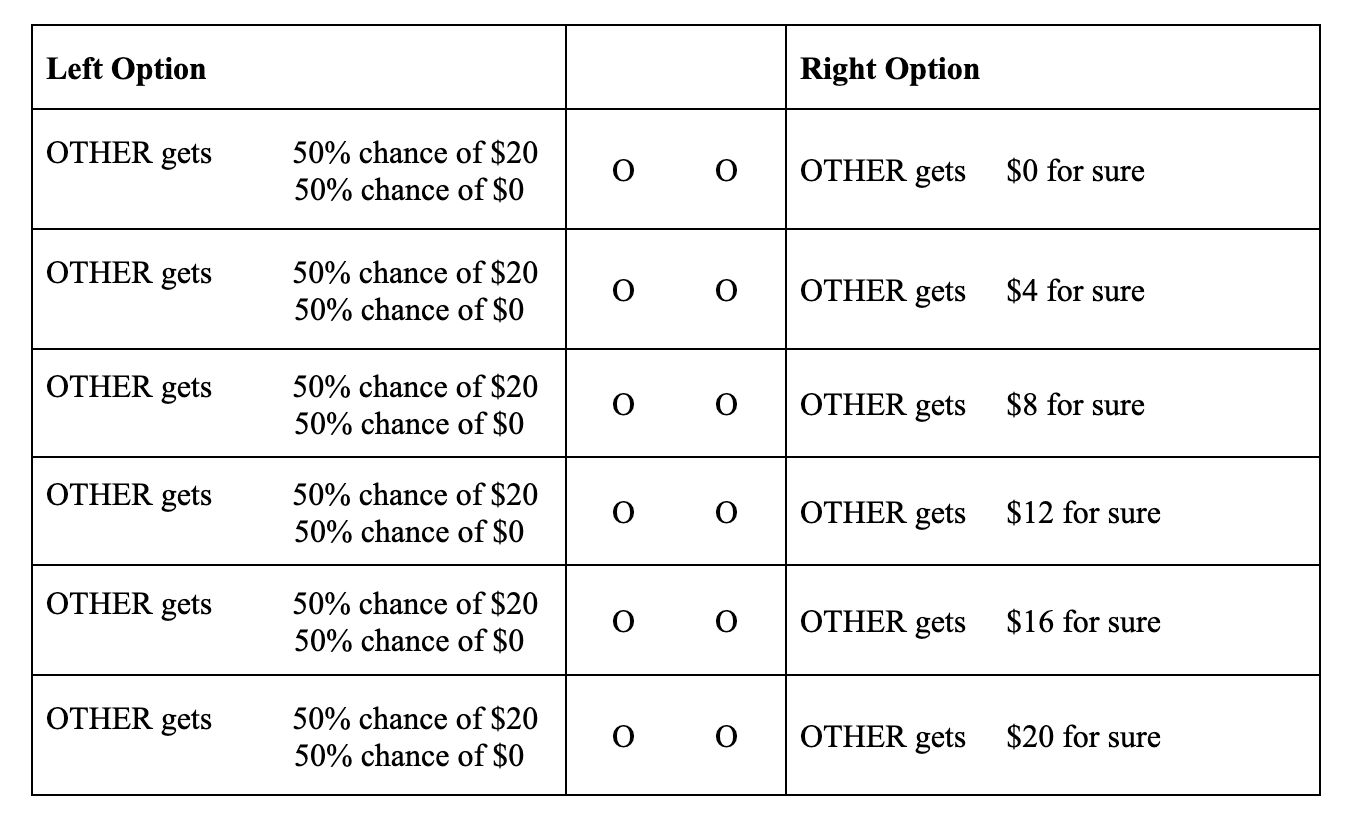}
\end{figure}

\noindent NEXT\\

\subsection*{INTERFACE PRACTICE QUESTION 2:}

\noindent Please implement the following selection: choose the lottery for the OTHER in all rows.

\begin{figure}[h!]
    \centering
    \includegraphics[scale=0.45]{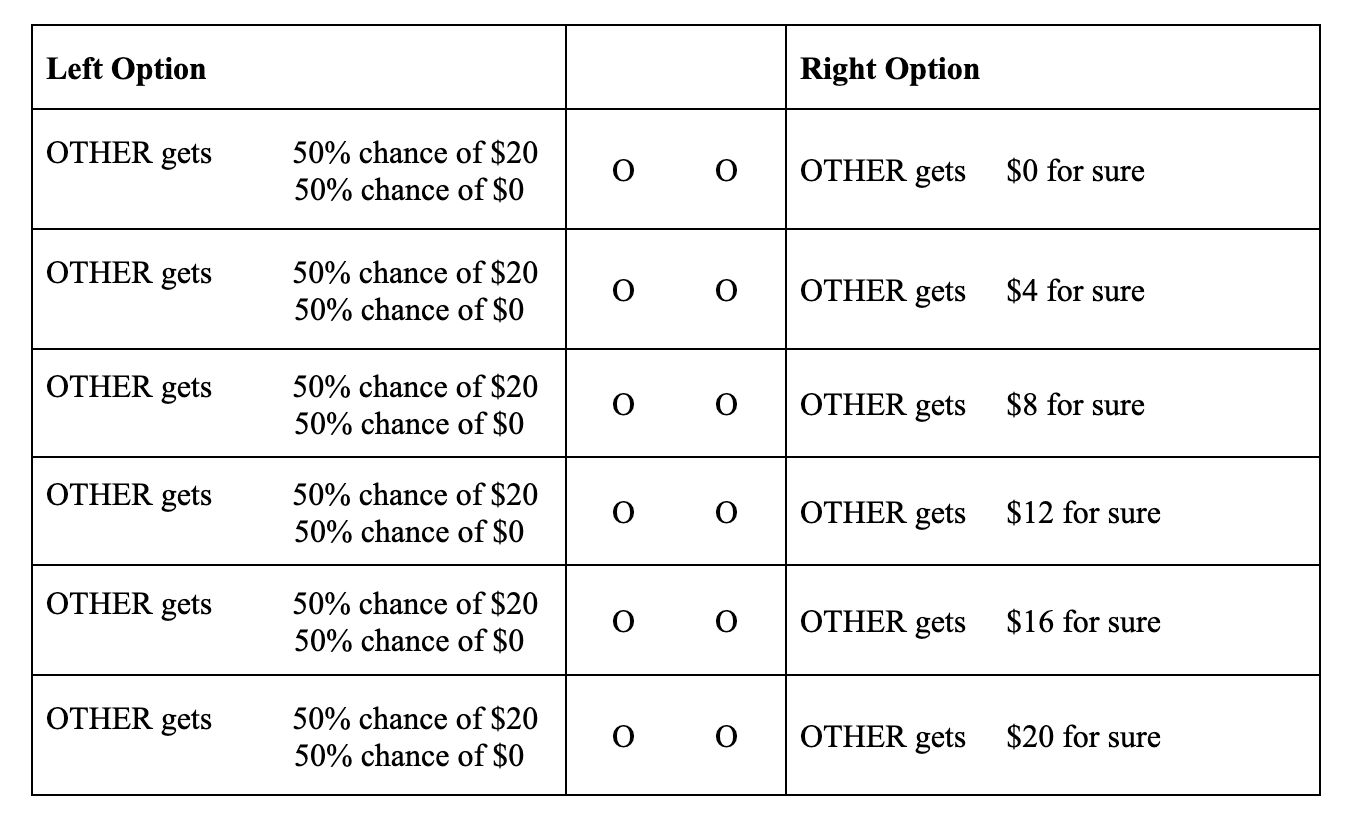}
\end{figure}

\noindent NEXT\\